\documentclass[12pt,a4paper]{article}
\usepackage[margin=1in]{geometry}
\usepackage{setspace}
\usepackage{authblk}

\usepackage{times}
\usepackage{bm}
\usepackage{natbib}

\usepackage{amsthm}
\usepackage{amssymb,amsmath}

\usepackage{algorithm}
\usepackage{multirow}
\usepackage{natbib}
\usepackage{xcolor}
\usepackage{graphicx}
\usepackage{subfigure}
\usepackage[font=small]{caption}

\newtheorem{theorem}{Theorem}[section]
\newtheorem{lemma}[theorem]{Lemma}
\newtheorem{proposition}[theorem]{Proposition}
\newtheorem{corollary}[theorem]{Corollary}

\newcommand{\Expect}[1]{E \left\{{#1}\right\}}
\newcommand{\Expects}[2]{E_{{#1}} \left\{{#2}\right\}}
\newcommand{\Var}[1]{\mathrm{var} \left\{{#1}\right\}}

\newcommand{\Abs}[1]{\left\vert{#1}\right\vert}
\newcommand{\Norm}[1]{\left\vert\left\vert{#1}\right\vert\right\vert}

\newcommand{\alphabar}{\overline{\alpha}}
\newcommand{\pihat}{\hat{\pi}}

\newcommand{\CSrm}[1]{}

\begin{document}

\title{Particle Metropolis-adjusted Langevin algorithms}
\author{Christopher Nemeth, Chris Sherlock and Paul Fearnhead}
\affil{Department of Mathematics and Statistics, Lancaster University, Lancaster LA1 4YF U.K. \\c.nemeth@lancaster.ac.uk, c.sherlock@lancaster.ac.uk, p.fearnhead@lancaster.ac.uk}

\maketitle

\begin{abstract}

This paper proposes a new sampling scheme based on Langevin dynamics 
that is applicable within pseudo-marginal and
particle Markov chain Monte Carlo algorithms. We investigate this algorithm's theoretical properties under standard asymptotics, which correspond to an increasing dimension of the
parameters, $n$. Our results show that the behaviour of the algorithm depends crucially on how accurately one can estimate the gradient of the log target density.
If the error in the estimate of the gradient is not sufficiently controlled as dimension increases, then asymptotically there will be no advantage over the simpler random-walk algorithm. However, if the error is sufficiently well-behaved, then the optimal scaling of this algorithm will be $O(n^{-1/6})$ compared to $O(n^{-1/2})$ for the random walk. Our theory also gives guidelines on how to tune the number of Monte Carlo samples in the likelihood estimate and the proposal step-size.

\end{abstract}


\section{Introduction}
\label{sec:introduction}

Markov chain Monte Carlo algorithms are a popular and well-studied methodology that can be used to draw samples from posterior distributions. 
 Over the past few years these algorithms have been extended to tackle problems where the model likelihood is intractable \citep{Beaumont2003}. \cite{Andrieu2009} showed that within the Metropolis--Hastings algorithm, if the likelihood is replaced with an unbiased estimate, then the sampler still targets the correct stationary distribution. \cite{Andrieu2010} extended this work further to create a class of Markov chain algorithms that use sequential Monte Carlo methods, also known as particle filters. 

Current implementations of pseudo-marginal and particle Markov chain
Monte Carlo use random-walk proposals to update the parameters
\cite[e.g.,][]{Golightly2011,Knape/deValpine:2012} and shall be
referred to herein as particle random-walk Metropolis
algorithms. Random walk-based algorithms propose a new value from
  some symmetric density centred on the current value. This
  density is not
  informed by the local properties of the posterior; however, we can
often obtain further information about such properties as we obtain our Monte Carlo estimate of the
posterior density, and at little or no additional
computational overhead. It is therefore natural to consider whether we can use this information to make better proposals for the parameters. In this paper we focus on using Monte Carlo methods to estimate the gradient of the log posterior density, and then use this to guide the proposed parameters towards regions of higher posterior probability. This results in a Monte Carlo version of the Metropolis-adjusted Langevin algorithm \citep{Roberts1998}, which we refer to herein as the particle Langevin algorithm.

When the likelihood is tractable, the Metropolis-adjusted Langevin algorithm has better theoretical
properties than the random-walk Metropolis algorithm. The mixing properties of these algorithms have been studied in the asymptotic limit as the dimension of the 
parameters, $n$, increases. In this asymptotic regime, the optimal
proposal step-size scales as $n^{-1/2}$ for the random-walk algorithm, but as $n^{-1/6}$ for the Metropolis-adjusted Langevin algorithm; and the optimal asymptotic acceptance rate is higher; see \cite{RobertsGelmanGilks1997}, \cite{Roberts1998} and
\cite{Roberts2001} for more details. It is natural to ask whether these advantages of the Metropolis-adjusted Langevin algorithm over the random-walk algorithm extend to pseudo-marginal and particle Markov chain Monte Carlo algorithms, and, in particular, how they are affected when only noisy estimates of the gradient of the log posterior density are available.

We investigate the asymptotic properties of the particle Langevin algorithm and show that its behaviour depends crucially on the accuracy of the estimate of the gradient of the log posterior density as $n$ increases. If the error in the estimate of a component of the gradient does not decay with $n$, then there will be no benefit over the particle random-walk algorithm. If the error is sufficiently well-behaved, then we find that the particle Langevin algorithm inherits the same asymptotic characteristics as the Metropolis-adjusted Langevin algorithm. The optimal proposal scales as $n^{-1/6}$, rather than $n^{-1/2}$, and there is a higher optimal acceptance rate. In this well-behaved regime we find that the number of particles should be chosen so that the variance in the estimate of the log posterior density is approximately 3. 

Furthermore, we provide explicit guidance for tuning the scaling of
the particle Langevin algorithm by aiming for a particular acceptance
rate. We show that the optimal acceptance rate depends crucially on
how accurately we estimate the log posterior density, a feature that is common
to other particle Markov chain Monte Carlo algorithms. As such, tuning
the particle Langevin algorithm using the acceptance rate is only
appropriate if we have an estimate of the variance of our estimator of
the log posterior density. Additionally, the optimal acceptance rate depends
on the accuracy of the gradient estimate. We propose a criterion
for choosing an appropriate scaling for the proposal given a fixed but
arbitrary number of particles. We provide an acceptance rate to tune to, which is
a function of the variance in the log posterior density estimate. This
acceptance rate is robust to the unknown accuracy of our estimate of
the gradient. Tuning to it will lead to an efficiency of at least
90\% of the efficiency of the optimally-scaled particle
Langevin algorithm, with the same number of particles and known accuracy of the gradient estimate. Under this criterion, and with sufficient particles so that the variance of the estimate of the log posterior density is approximately 3, we should scale the step-size so that the acceptance rate is 11\%.

\section{Efficient Markov chain Monte Carlo with intractable likelihoods}
\label{sec:part-marg-metr}

Let $p(z \mid x)$ be a model likelihood, with data $z \in
\mathcal{Z} \subseteq \mathbb{R}^{n_z}$ and model parameters
$x \in \mathcal{X} \subseteq \mathbb{R}^n$. Using Bayes' theorem, the posterior density over the parameters, up to a constant of proportionality, is $\pi(x) \propto p(z\mid x)p(x)$, where $p(x)$ is a prior density for $x$. 

Markov chain Monte Carlo algorithms draw samples,
$(x_1,\ldots,x_J)$, from the posterior
distribution. Typically, these samples are generated using the Metropolis--Hastings algorithm, where proposed values $y$ are sampled from a proposal distribution $q(\cdot\mid x)$ and accepted  with probability 
\begin{equation}
\label{eq:13}
  \alpha(y\mid x) = \min\left\{1, \frac{p(z\mid y)p(y)q(x\mid y)}{p(z\mid x)p(x)q(y\mid x)}\right\}.
\end{equation}

The Metropolis--Hastings algorithm requires that the likelihood $p(z\mid x)$ be tractable, but there are many situations where it can only be evaluated approximately. \cite{Andrieu2009} and \cite{Andrieu2010} have shown that the Metropolis--Hastings algorithm is still valid in this setting, provided there is a mechanism for simulating unbiased, 
non-negative estimates of the likelihood. This technique is known as pseudo-marginal Markov chain Monte Carlo.

The pseudo-marginal approach presupposes that a non-negative unbiased estimator $\hat{p}(z\mid x,\mathcal{U}_x)$ of $p(z\mid x)$ is available, where $\mathcal{U}_x \sim p(\cdot\mid x)$ denotes the random variables used in the sampling mechanism to generate an estimate of the likelihood. 
We then define a target density on the joint space $(x,\mathcal{U}_x)$ as,
\begin{equation}
  \label{eq:14}
  \hat{\pi}(x,\mathcal{U}_x) \propto \hat{p}(z\mid x,\mathcal{U}_x)p(\mathcal{U}_x\mid x)p(x).
\end{equation}
Since the estimate is unbiased, the marginal density of $x$ is
\begin{eqnarray*}
  \int   \hat{\pi}(x,\mathcal{U}_x) d\mathcal{U}_x &\propto& \int \hat{p}(z\mid x,\mathcal{U}_x)p(\mathcal{U}_x\mid x)p(x) d\mathcal{U}_x = p(z\mid x)p(x),
\end{eqnarray*}
which is the posterior density of interest. 

A valid Markov chain Monte Carlo algorithm targeting \eqref{eq:14}, with proposal $q(y\mid x)p(\mathcal{U}_y\mid y)$, has acceptance probability of the form \eqref{eq:13}, but with the intractable likelihoods, $p(z\mid x)$ and $p(z\mid y)$ replaced with realizations from their unbiased estimators, $\hat{p}(z\mid x,\mathcal{U}_x)$ and  $\hat{p}(z\mid y,\mathcal{U}_y)$. 

The efficiency of the Metropolis--Hastings algorithm is highly dependent on the choice of proposal distribution $q(y\mid x)$. Ideally, the proposal would use local information about the posterior to sample from areas of higher posterior density. 
One such proposal is the Metropolis-adjusted Langevin algorithm \cite[]{Roberts1998} which incorporates the gradient of the log posterior density, $\nabla \log \pi(x)$, within the proposal. 
The asymptotic behaviour of this algorithm, as the
number of parameters, $n$, increases, gives an optimal step-size of $O(n^{-1/6})$ \citep{Roberts1998} compared to $O(n^{-1/2})$ for the random-walk Metropolis algorithm \citep{RobertsGelmanGilks1997}. As a result, to maintain a reasonable acceptance rate for large
$n$, the Metropolis-adjusted Langevin algorithm may propose larger jumps in the posterior than the random-walk Metropolis algorithm,
reducing the first order auto-correlation and improving the mixing of the Markov chain. 

Using the Metropolis-adjusted Langevin algorithm in the pseudo-marginal setting is challenging because if the likelihood is intractable then typically, $\nabla \log \pi(x)$ will also be intractable. Therefore, one needs to efficiently estimate both the posterior density $\hat{\pi}(x)$, and its log gradient $\hat{\nabla} \log \pi(x)$. The resulting
algorithm, which we call the particle Langevin algorithm, proposes a new parameter value $y$ as
\begin{equation}
\label{eq:pmala}
  y = x + \lambda Z + \frac{\lambda^2}{2} \hat{\nabla} \log \pi(x), \quad Z \sim \mathcal{N}(0,\mathrm{I}).
\end{equation}
It is often possible to generate a Monte Carlo estimate of the gradient of the log posterior density with little additional computational overhead, from the output
of the same Monte Carlo method used to estimate the likelihood \citep{Poyiadjis2011}.
The efficiency of the particle Langevin algorithm will depend on the choice of scaling parameter $\lambda$ and the accuracy of the estimator $\hat{\nabla} \log \pi(x)$. In the next section we derive asymptotic results which allow us to optimally choose $\lambda$ and which show how the efficiency of the particle Langevin algorithm depends on the accuracy of the estimator of the gradient.

\section{Theoretical results}
\label{sec:optim-scal-results}

\subsection{High-dimensional setting}
\label{sec:high-d}

In this section we present two key theoretical results and investigate their practical consequences. These results apply in the general pseudo-marginal setting, but the practical guidance requires specific
distributional assumptions and is specific to algorithms where the estimate of the
likelihood is obtained using a particle filter. For simplicity,
therefore, we continue to use particle Langevin as a
general term for both pseudo-marginal and particle Markov chain Monte Carlo algorithms. All proofs are presented in the Supplementary Material. 

We consider an infinite sequence of targets
$\pi^n(x^n),~n=1,\dots$, where $x^n$ is an $n$-dimensional vector. We obtain limiting forms for the acceptance rate and
expected squared jump distance, $J_n$, for the particle Langevin proposal. The expected squared jumping distance has been used extensively as a measure of mixing of
Markov chain Monte Carlo algorithms \cite[e.g.,][]{BeskosRobertsStuart:2009,Sherlock/Roberts:2009,Sherlock:2013}, where maximizing it is
equivalent to minimizing the first order auto-correlation of the Markov chain. The particle Metropolis-adjusted Langevin kernel itself is not a positive operator; however any kernel with a rejection probability of at least $0.5$ for all possible moves is a positive operator and typically we will be tuning our algorithm to give an average acceptance probability between $0.1$ and $0.15$, so that most moves have a rejection probability in excess of $0.5$. Moreover, in the presence of a limiting diffusion, the limiting, scaled expected squared jumping distance is the speed of the diffusion and hence precisely the right measure of efficiency. \cite{Sherlock:2015} and \cite{Roberts1998} show limiting diffusions, respectively, for the particle random-walk Metropolis algorithm and the Metropolis-adjusted Langevin algorithm, suggesting the likely existence of a limiting diffusion for the particle Metropolis-adjusted Langevin kernel. 

We start by considering the idealized particle Langevin algorithm where, for any given $x^n$,  an unbiased stochastic estimate of the target density is used, with an exact gradient of the log target density, $\nabla\log \pi^n(x^n)$. This algorithm is unlikely to be usable in practice, but provides a useful reference point for the more general particle Langevin proposal where we assume that we have a noisy and possibly biased estimate of $\nabla \log \pi^n(x^n)$. Introducing the possibility of both noise and bias in the estimate allows our results to be applied to a wider range of algorithms that could be used to estimate the gradient of the log target density.

We study a target of the form
\begin{equation}
\label{eqn.product.form}
\pi^n(x^n)=\prod_{i=1}^nf(x_i^n),
\end{equation}
where $x_i^n$ denotes the $i\mbox{th}$ component of an $n$ dimensional vector
$x^n$. 
We set $g(x)=\log f(x)$ and assume that $g(x)$ and its derivatives $g^{(i)}(x)$ satisfy
\begin{equation}
\label{eqn.poly.mom.g}
|g(x)|,|g^{(i)}(x)|\le M_0(x), \quad i=1,\dots,8,
\end{equation}
where $M_0(x)$ is some polynomial, and 
\begin{equation}
\label{eqn.finite.moments}
\int_{\mathbb{R}}x^kf(x)~dx<\infty, \quad k=1,2,3,\dots.
\end{equation}

Our assumptions on the form of the target, \eqref{eqn.product.form}--\eqref{eqn.finite.moments}, are the same as those in
\cite{Roberts1998}. In particular, for tractability, the target is assumed to have a product form. This apparently restrictive assumption is common in much of the literature on high-dimensional limit results, including
\cite{RobertsGelmanGilks1997}, \cite{Roberts1998}, \cite{NealRoberts:2006}, \cite{RobRosSimTem} and \cite{Sherlock:2015}. Some of these results have been
 extended to more general settings 
 \cite[e.g.,][]{Roberts2001,Bedard:2007,Sherlock/Roberts:2009,BeskosRobertsStuart:2009,Sherlock:2013}, where 
 optimality criteria obtained using a product target
 have been found to hold for more general statistical applications. The results are also widely used within 
 adaptive Markov chain Monte Carlo algorithms
 \cite[e.g.,][]{AndrieuThoms:2008,RobRos:2009,Sarkkaetal:2015}.

We consider the additive noise in the log target density at the current and
proposed values:
\[
W^{n}=\log \pihat^n(x^n,\mathcal{U}_x^{n})-\log \pi^n(x^n),
~~~
V^{n}=\log \pihat^n(y^n,\mathcal{U}_y^{n})-\log \pi^n(y^n)
\]
and their difference
\begin{equation}
\label{eq:B}
B^{n}=V^{n}-W^{n}.
\end{equation}
As in \cite{Pitt2012}, \cite{Sherlock:2015}  and \cite{Doucet:2015}, we
assume that the distributions of $V^{n}$ and $W^{n}$ are independent
of position. This is unlikely to hold
in practice, but simulations in those articles show that it can hold
approximately and that guidance obtained from the resulting theory can be robust to
 variations with position. In the Supplementary Material we investigate and discuss this assumption for the scenarios in
Section \ref{sec:particle-filtering}.

For particle Markov chain Monte Carlo, \cite{Berard2014} examine the particle filter in the limit of a large
number, $N$, of particles acting on a large number of observations and find
that 
\begin{equation}
\label{eqn.SARa}
V^n\mid x^n,y^n,w~~\sim~~ \mathcal{N}\left(-\frac{1}{2}\sigma^2,\sigma^2\right),
\end{equation}
for some fixed $\sigma^2\propto 1/N$. From the definition of $W^n$, and directly from \eqref{eq:14}, it follows that
\begin{equation}
\label{eqn.targ.W}
W^n\sim \mathcal{N}\left(\frac{1}{2}\sigma^2,\sigma^2\right),
~~~B^n\sim \mathcal{N}(-\sigma^2,2\sigma^2),
\end{equation}
when the chain is at stationarity \citep{Pitt2012}.

We apply our theoretical results to this common scenario with the
assumption \cite[e.g.,][]{Pitt2012,Sherlock:2015,Doucet:2015} that the
computational cost is proportional to $N$ and hence
inversely proportional to $\sigma^2$. 
Therefore, our measure of efficiency is, up to a
constant of proportionality,
\begin{equation}
\label{eqn.eff.CPU}
\mbox{Eff}(\ell,\sigma^2)= \sigma^2 J_n(\ell,\sigma^2),
\end{equation}
where $\ell$ is related to the scaling of the proposal as
in \eqref{eqn.define.lambda} and Theorem \ref{sec:scaling-grads}.

We consider a range of levels of control for the bias and variance of the errors in the
estimate of each component of the gradient. For a given level of control, we investigate the scaling that is necessary to achieve a non-degenerate limiting
acceptance rate, and the behaviour of the efficiency function in
that limit. A natural corollary of our analysis is that these scaling requirements, and the resulting general forms for the limiting acceptance rate and expected squared jump distance, would persist even
if we were able to use an unbiased estimate of the gradient.

\subsection{Idealized particle Langevin algorithm}
\label{sec:pseudo-marginal-mala}

In this section we consider the idealized particle Langevin algorithm, providing general limiting forms for the
acceptance rate and expected squared jump distance. 

Let the scaling for the proposal on the target $\pi^n$ be
\begin{equation}
\label{eqn.define.lambda}
\lambda_n=\ell n^{-1/6},
\end{equation}
where $\ell>0$ is a tuning parameter. As mentioned earlier, in the idealized particle Langevin algorithm we make the unrealistic assumption that the gradient of the log target density may be evaluated precisely so that the $i\mbox{th}$ component of the proposal is
\begin{equation}
\label{eqn.define.Y}
Y_i^n=x_i^n+\lambda_nZ_i+\frac{1}{2}\lambda_n^2g'(x_i^n),
\end{equation}
with $Z_i\sim \mathcal{N}(0,1)~(i=1,\dots,n)$ 
independent of all other sources of variation. 

Let $\alpha_n(x,w;y,v)$ be the acceptance probability for the
idealized particle Langevin algorithm with current value $(x,w)$ and proposed value $(y,v)$. 
We are interested in the expected acceptance rate and the expected squared jump distance,
\begin{eqnarray*}
\overline{\alpha}_n(\ell)&=&\Expect{\alpha_n(X^n,W^n;Y^n,V^n)},\\
J_n(\ell)&=&\Expect{\Norm{Y^n-X^n}^2\alpha_n(X^n,W^n;Y^n,V^n)},
\end{eqnarray*}
where expectation is over $X^n,W^n,Y^n,V^n$  with distributions as defined
in \eqref{eqn.product.form}, \eqref{eqn.define.Y}, \eqref{eqn.SARa}  and \eqref{eqn.targ.W}. Our first result is as follows.

\begin{theorem}
\label{thrm.main}
As $n\rightarrow \infty$, the following limits hold in probability:
\[
\alphabar_n(\ell)\rightarrow
\alpha(\ell)=2 E \left\{
\Phi\left(\frac{B}{\ell^3K}-\frac{\ell^3K}{2}\right)
\right\},
~~~
n^{-2/3}J_n(\ell) \rightarrow \ell^2\alpha(\ell),
\]
where in distribution $B = \lim_{n\rightarrow \infty}B^n$, and $B^n$
is defined in \eqref{eq:B}. Here,
\begin{equation}
  \label{eq:K}
K=\left[\frac{1}{48}\Expect{5g'''(X)^2-3g''(X)^3}\right]^{1/2} \in \mathbb{R}^+,  
\end{equation}
where expectation in \eqref{eq:K} is with respect to the density
$f(x)$.
\end{theorem}

The following corollary details the parameters that optimize the
efficiency function for the particle filter scenario. 
\begin{corollary}
\label{sec:opt-acc-rate}
Subject to \eqref{eqn.SARa} and \eqref{eqn.targ.W}, the efficiency defined in \eqref{eqn.eff.CPU} 
is maximized when the scaling and noise variance are $\ell_{\mathrm{opt}} \approx 1.125K^{-1/3}$, and $\sigma_{\mathrm{opt}}^2 \approx 3.038$, at which point, the limiting acceptance rate is $\alpha_{\mathrm{opt}}\approx 15.47\%$.
\end{corollary}

The optimal variance of the noise in the log target density differs only
slightly from that of the particle random-walk Metropolis algorithm, where
$\sigma^2_{\mathrm{opt}}\approx 3.283$ \cite[]{Sherlock:2015}; however, the optimal
asymptotic acceptance rate is increased from $7.00\%$ to $15.47\%$ and the scaling is improved from $O(n^{-1/2})$ to $O(n^{-1/6})$. Therefore, for large $n$, the particle Langevin algorithm permits larger jumps leading to a more efficient proposal distribution.

\begin{figure}[t!]
  \centering
  \includegraphics[scale=0.5]{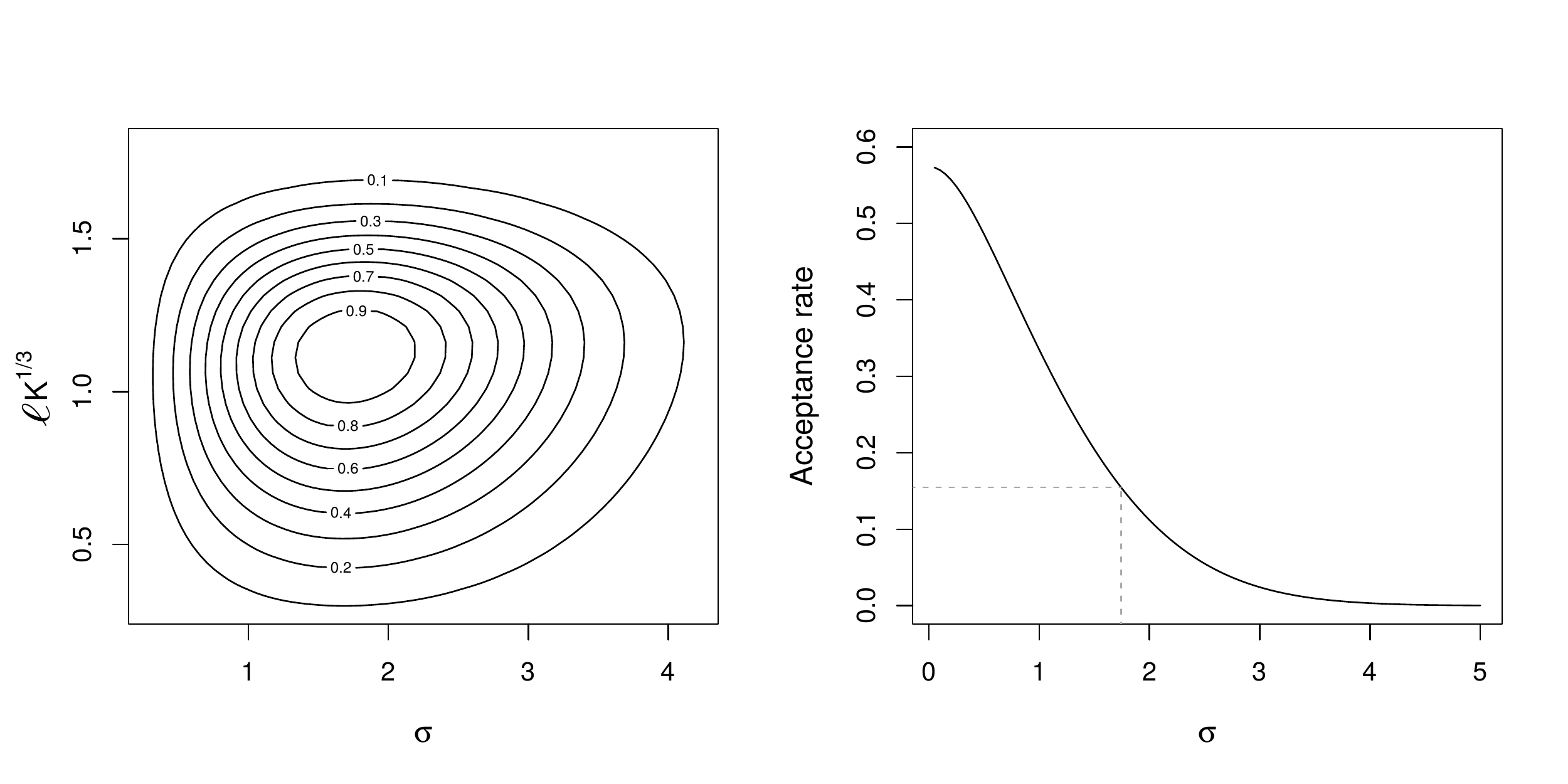}
  \caption{Contour plot of the relative efficiency Eff($\ell$,$\sigma^2$)/Eff($\ell_{\mathrm{opt}}$,$\sigma^2_{\mathrm{opt}}$) (left panel), and asymptotic acceptance rate (right panel) plotted against $\sigma$, where $\ell$ is optimized for each $\sigma$, for the
  idealized particle Langevin algorithm and for the particle Langevin algorithm in asymptotic regime (3) of Theorem \ref{sec:scaling-grads}.}
  \label{fig:contours}
\end{figure}

Figure \ref{fig:contours} shows the relative efficiency as a function
of the scaling and the standard deviation, $\sigma$, of the noise, and
the optimal acceptance rate as a function of $\sigma$. The left panel
shows that over a wide range of variances the optimal scaling $\ell$
is close to $1.125K^{-1/3}$, and over a wide range of scalings, the
optimal variance $\sigma^2$ is close to $3.038$. This relative
insensitivity between the scaling and variance means that the scaling
which maximizes the expected squared jump distance over all possible
noise variances will be close to the optimal scaling for any specific
noise variance in a large range. The right panel gives the acceptance
rate for a range of variances, where $\ell$ is optimally tuned for each $\sigma^2$. The optimal acceptance rate varies considerably over a range of sensible noise variances. This suggests that, given a sensible, but not necessarily optimal noise variance, tuning to achieve an acceptance rate of about $15\%$ may lead to a relatively inefficient algorithm. Instead, one should either choose a scaling which optimizes the effective sample size directly, or estimate the variance in the noise in the log target density, find the acceptance rate that corresponds to the optimal scaling conditional on the estimated variance, and tune to this.

\subsection{Scaling conditions for the particle Langevin algorithm}
\label{sec:scal-grad-comp}

In the particle Langevin algorithm we do not have an exact estimate
for the gradient of the log target density. In fact, depending on the approach
used to estimate the gradient, the estimate may be both biased and
noisy. In this section we give conditions on the bias and noise of the
gradient estimate that would lead to an efficient proposal distribution. 

We start by considering the $i\mbox{th}$ component of the particle Langevin
proposal ($i=1,\dots,n$): 
\begin{equation}
\label{eqn.define.Y.bias}
Y_i^n=x_i^n+\lambda_nZ_i+\frac{1}{2}\lambda_n^2 \left[g'(x_i^n)+\frac{1}{n^\kappa}\left\{b(x_i^n) + \tau U_{x_i^n}\right\} \right],
\end{equation}
where, for all $i$, $Z_i\sim \mathcal{N}(0,1)~(i=1,\dots,n)$ and $U_{x_i^n}$ are
independent of each other and of all other sources of variation. For
any $x$, $U_x$ is a random variable with a distribution that is
independent of $X$ and $W$, with $E(U_{x})=0$, 
$\mathrm{var}(U_{x})=1$ and
\begin{equation}
\label{eqn.mom.U}
\Expect{\Abs{U_{x}}^k}<\infty, \quad k>0.
\end{equation}
In the Supplementary Material, the assumption that the variance of $U_{x_i^n}$ is constant, and that
$U_{x_i^n}$ and $W$ are independent, are checked on the models from Section \ref{sec:particle-filtering}; the variance is shown to change by at most an order
of magnitude, and independence is shown to be a good working assumption.

Even though the variance of the noise is fixed, the bias $b(x_i^n)$
in the estimate of the $i\mbox{th}$ component of the gradient (at $x_i^n$) can be position specific. Furthermore, we assume that $b(x)$ and its derivatives $b^{(i)}(x)~(i=1,\dots,7)$ satisfy
\begin{equation}
  \label{eq:cond.b}
|b(x)|,|b^{(i)}(x)|\le M_0(x),  
\end{equation}
where $M_0(x)$ is, without loss of generality, the
same polynomial as in \eqref{eqn.poly.mom.g}.

The particle Langevin proposal \eqref{eqn.define.Y.bias} can be viewed as
a generalization of the Metropolis-adjusted Langevin proposal, which can be retrieved by setting $b(x)=\tau=0$. The bias and noise
components of \eqref{eqn.define.Y.bias} are scaled by an $n^{-\kappa}$
term, where $\kappa \geq 0$. If $\kappa=0$, as shall be shown
in Part (1) of Theorem \ref{sec:scaling-grads}, in order to achieve a
non-degenerate limiting acceptance rate, the scaling of the proposal
must be chosen so that the particle Langevin proposal has the same
limiting behaviour as the particle random-walk algorithm. In addition to
the definition of $K$ in \eqref{eq:K}, we define
\begin{eqnarray}
\label{eqn.defn.Kstar}
K_*^2&=&\Expects{f}{b(X)^2}+\frac{1}{2}\tau^2,\\
\label{eqn.defn.Kstarstar}
K_{**}&=&-\frac{1}{4}\Expects{f}{b'(X)g''(X)}.
\end{eqnarray}
 where, by assumptions \eqref{eqn.poly.mom.g}, 
\eqref{eqn.finite.moments} and \eqref{eq:cond.b}, these expectations
are finite.

\begin{theorem}
\label{sec:scaling-grads}
Define $\psi(a;B)=\Phi(B/a - a/2)$, where in distribution $B = \lim_{n\rightarrow \infty}B^n$, and $B^n$
is defined in \eqref{eq:B}. As $n \rightarrow \infty$ the following limits hold in probability:

(1) If $\kappa=\frac{1}{3}-\epsilon$, where $0<\epsilon \leq \frac{1}{3}$, then $\lambda_n=\ell n^{-1/6-\epsilon}$ for a non-degenerate limiting acceptance rate, whence  
\[
\alphabar_n(\ell)\rightarrow \alpha^{(1)}(\ell) =
2 E \left\{\psi(\ell K_*;B)\right\},
~~~
n^{-1+\kappa}J_n(\ell) \rightarrow \ell^2\alpha^{(1)}(\ell).
\]

(2) If $\kappa=\frac{1}{3}$, then $\lambda_n=\ell n^{-1/6}$ for a non-degenerate limiting acceptance rate, whence  
\[
\alphabar_n(\ell)\rightarrow \alpha^{(2)}(\ell)=
2 E \left[\psi\left\{\left(\ell^6K^2+2\ell^4K_{**}+\ell^2K_*^2\right)^{1/2};B\right\}\right],~~~ n^{-2/3}J_n(\ell) \rightarrow \ell^2\alpha^{(2)}(\ell);
\]
\[
\mbox{where}~~~\ell^6K^2+2\ell^4K_{**}+\ell^2K_*^2\ge 0.
\]

(3) If $\kappa > \frac{1}{3}$, then $\lambda_n=\ell n^{-1/6}$ for a non-degenerate limiting acceptance rate, whence  
\[
\alphabar_n(\ell)\rightarrow \alpha^{(3)}(\ell) =2 E \left\{
\psi\left(\ell^3K;B\right)
\right\},
~~~
n^{-2/3}J_n(\ell) \rightarrow \ell^2\alpha^{(3)}(\ell).
\]
\end{theorem}
The theorem highlights the relative
contributions of the change in the true
posterior and the error in the
gradient term appearing in the Metropolis--Hastings acceptance ratio. When $\kappa<1/3$, the contribution from the gradient
term must be brought under control by choosing a smaller
scaling, but when this smaller scaling is used, 
the limiting acceptance ratio for the Metropolis-adjusted Langevin algorithm is $1$ and so the
roughness of the target itself, $K$, is irrelevant. It is
only at the most
extreme end of regime (1), when $\kappa=0$, that the expected squared jump distance is of the same
order of magnitude as for the pseudo-marginal random-walk Metropolis algorithm. By contrast, when the scaling is $\kappa>1/3$ the effect
of the errors in the gradient on the acceptance ratio  is negligible;
the behaviour is that of the idealized particle Langevin algorithm. 
The case where $\kappa = 1/3$ gives a balance between the
contributions to the acceptance ratio.

\subsection{Tuning the particle Langevin algorithm}
\label{sec:tuning-particle-mala}

Theorem \ref{sec:scaling-grads} has two important implications. Firstly, it provides insight into
the relative performance of the particle Langevin algorithm compared to the particle random-walk algorithm. Asymptotically, the former algorithm has better
mixing properties providing there is some control over the error of each component of the gradient, $\kappa>0$. The greater the control, the better the scaling of the step-size as the number of parameters increases. 
Under our assumption on the target (\ref{eqn.product.form}) it would be natural to expect that condition (3) of Theorem \ref{sec:scaling-grads} would hold,
and that the optimal scaling would be proportional to $n^{-1/6}$.
This is because, for the particle Langevin algorithm, we need to control the variance of the estimate of the posterior density as $n$ increases. This requires the number of particles used to estimate each
component of (\ref{eqn.product.form}) to increase linearly with $n$ so that the Monte Carlo variance of the  estimate of each term in the product (\ref{eqn.product.form}) is of order $1/n$. Under this regime, the Monte Carlo error of the estimate of each component 
of the gradient would be of order $n^{-1/2}$, which corresponds to $\kappa=1/2$. Empirical investigations for two models reported in Section \ref{sec:particle-filtering}, and 
Section \ref{sect.regime.diags} of the Supplementary Material, 
indicate that both fit into  case (3). 

The second consequence of Theorem \ref{sec:scaling-grads} is the implementation guidance for the particle Langevin algorithm. In particular, results on optimal acceptance rates are important for tuning the proposal appropriately, and results on the expected squared jump distance aid in the choice of number of particles.

In the particle filter scenario, in an analogous manner to the first
part of the proof of Corollary \ref{sec:opt-acc-rate}, the three
acceptance rates can
be shown to simplify to 
\begin{eqnarray*}
\alpha^{(1)}(\ell,\sigma^2)&=&2\Phi\left\{-\frac{1}{2}(\ell^2K_*^2+2\sigma^2)^{1/2}\right\},\\
\alpha^{(2)}(\ell,\sigma^2)&=&2\Phi\left\{-\frac{1}{2}(\ell^6K^2+2\ell^4K_{**}+\ell^2K_*^2+2\sigma^2)^{1/2}\right\},\\
\alpha^{(3)}(\ell,\sigma^2)&=&2\Phi\left\{-\frac{1}{2}(\ell^6K^2+2\sigma^2)^{1/2}\right\},
\end{eqnarray*}
where we now make the dependence of the acceptance rates on the
distribution of the noise difference, $B$, explicit through the
parameter $\sigma^2$.
The $K_{**}$ (\ref{eqn.defn.Kstarstar}) term appearing in the acceptance rate for case (2) can be negative, and this can lead 
to the counter-intuitive situation where increasing the step-size can increase the acceptance rate; see the Supplementary Material.

For regime (3) the optimal variance ($\sigma^2\approx 3.0$) and
acceptance rate ($\alpha\approx 15.5$) are supplied by Corollary
\ref{sec:opt-acc-rate}. For regimes (1) and (2) the optimal choices
will depend on the relationship between the number of particles, $K_*$
and $K_{**}$, and this relationship is unknown. If $K_*$ were fixed
then the optimal variance for regime (1) would be $\sigma^2\approx
3.3$ as for the particle random-walk Metropolis algorithm, because the
efficiency has the same general form; however it is reasonable to
assume that $K_*$ and $K_{**}$ will decrease as $\sigma^2$
decreases. 
In this case, we can always slightly increase our efficiency measure (10) by reducing $\sigma^2$ and increasing $\ell$
in such a way that $\ell^2K^2_*+2\sigma^2$ is fixed. In a real, finite-dimensional, problem
our limit theory is not appropriate for very large $\ell$. However, the above argument suggests the optimal variance will be
less than 3.3. A similar argument holds in case (2), and thus we recommend choosing the number of particles
so that the variance of the estimate of the log target density is roughly 3.0.
Conditional on a choice of the number of particles or, equivalently, of the variance of the estimator of the log target density, it is possible to provide an acceptance rate that is close to optimal simultaneously across all three regimes. The scaling can therefore be adjusted to obtain this acceptance rate.  The idea is to maximize the worst-case performance of the particle Langevin algorithm.

Fix $\sigma^2$ and assume that the behaviour of the particle Langevin algorithm is
described by one of the limiting regimes of Theorem
\ref{sec:scaling-grads}. Denote the complete details of this regime by $r=(\kappa,K,K_{*},K_{**})\in\mathcal{R}$,
 where $\mathcal{R}$ denotes the set of possible regimes. Given the counter-intuitive behaviour described above when $K_{**}<0$, we consider only
regimes with $K_{**}\geq 0$. Denote the asymptotic expected
squared jump distance of the particle Langevin algorithm as $J(\alpha,r)$ for
regime $r$, where $\ell$ is chosen to give an average acceptance
probability $\alpha$. This is well-defined for
$0<\alpha<2\Phi(-\sigma/ \surd 2)$, as the acceptance rate is
continuous and monotonically decreasing with $\ell$. Then, for this
regime, the relative efficiency of the particle Langevin, with average acceptance rate $\alpha$, can be measured as 
\[
 \mbox{EffR}(\alpha,r)=\frac{J(\alpha,r)}{\max_{\alpha'}J(\alpha',r)},
\]
the ratio of the expected squared jump distances for this implementation of the particle Langevin algorithm and for the optimal implementation within regime $r$. A robust choice of average acceptance rate to tune to is the value that maximizes the minimum efficiency,
\[
 \arg\max_{\alpha} \min_{r} \mbox{EffR}(\alpha,r).
\]
We call this the maximin acceptance rate. Calculating, for any $\sigma$, 
the corresponding maximin acceptance rate is straightforward numerically. In Figure \ref{fig:minimax} we show the maximin acceptance rate as a function of $\sigma$ and the corresponding worst-case efficiency.
The maximin
choice of acceptance rate gives a worst-case relative efficiency of approximately 90\% for all values of $\sigma$. For $\sigma^2\approx 3$ ($\sigma \approx 1.73$)
we have a maximin optimal average acceptance rate of $\approx 11\%$.

\begin{figure}[t!]
  \centering
  \includegraphics[scale=0.35]{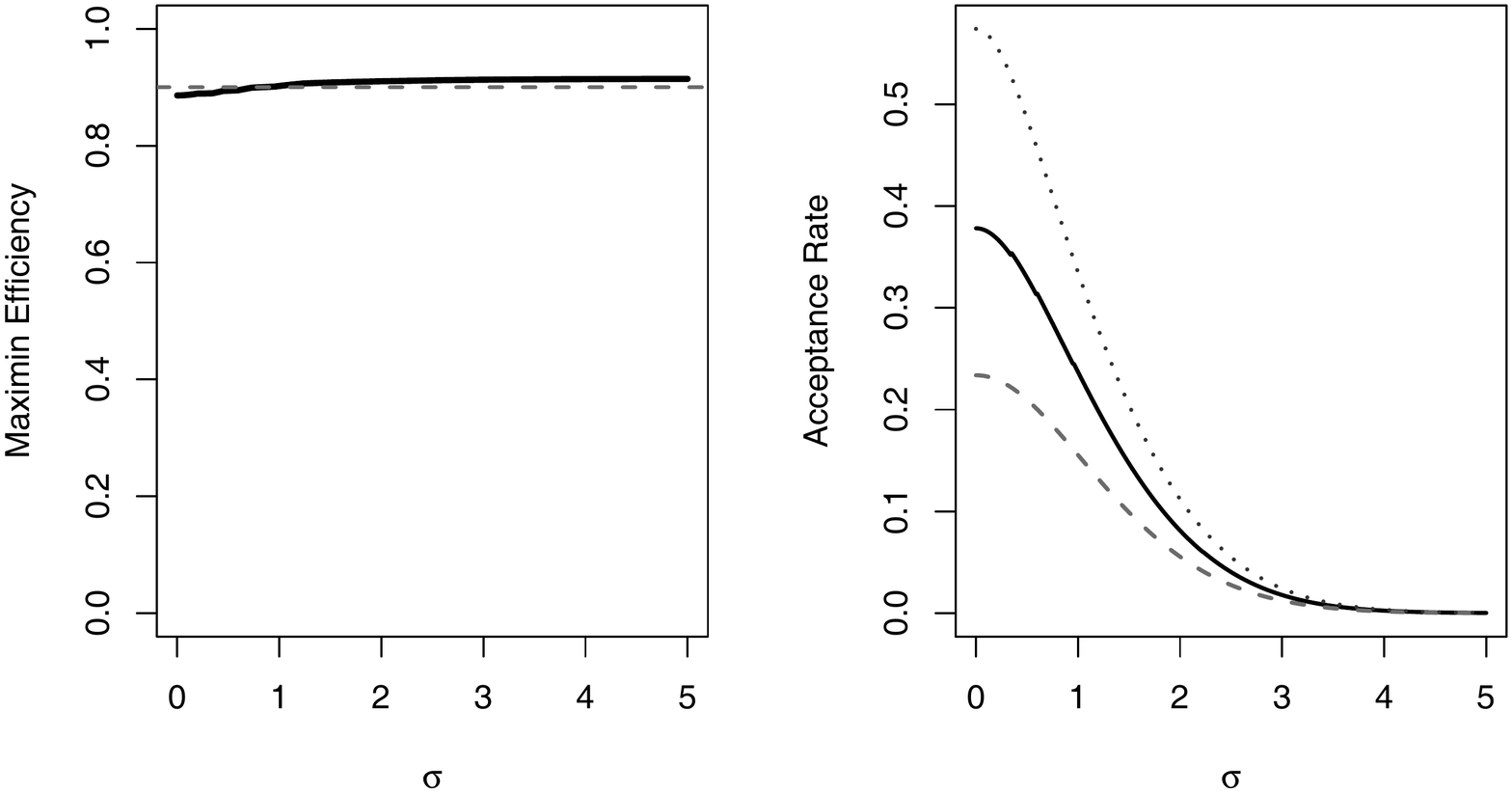}
  \caption{Plots of the maximin implementation. The left panel shows the maximin relative efficiency (black solid line) against the 90\% efficiency level (dashed line). The right panel shows the maximin optimal acceptance rates as a function of $\sigma$ (black). For comparison we also plot the optimal acceptance rate for regime (3), idealized particle Langevin algorithm, (dotted line); and regime (1) (dashed line). }
  \label{fig:minimax}
\end{figure}

\section{Inference for state-space models}
\label{sec:particle-filtering}

\subsection{Particle filtering}
\label{sec:particle-filtering-1}

In this section we apply the particle Langevin algorithm to two state-space model examples, where, for consistency with notation, we denote $x \in \mathcal{X}$ as a vector of model parameters and let $\{S_t: 1 \leq t \leq T\}$  be a latent Markov process taking values on some measurable space $\mathcal{S} \subseteq \mathbb{R}^{n_s}$. The process is fully characterized by its initial density $p(s_1\mid x)=\mu_x(s_1)$ and transition probability density
\begin{equation*}
  p(s_t\mid s_{1:t-1},x) = p(s_t\mid s_{t-1},x) = f_x(s_t\mid s_{t-1}),
\end{equation*}
where we use the notation $s_{1:t-1}$ in place of $(s_1,\ldots,s_{t-1})$.

We assume that the process $S_t$ is not directly observable, but partial observations are received via a second process $Z_t\subseteq \mathcal{Z}^{n_z}$. The observations $z_t$ are conditionally independent given $S_t$ and are defined by the probability density
\begin{equation*}
  p(z_t\mid z_{1:t-1},s_{1:t},x) = p(z_t\mid s_t,x) = g_x(z_t\mid s_t).
\end{equation*}
The posterior density of the parameters $\pi(x) \propto p(z_{1:T}\mid x)p(x)$, where $p(x)$ is a prior density for $x$,  is obtained by integrating out the latent process $\{S_t\}_{t \geq 1}$ to give the marginal likelihood
\begin{equation*}
  p(z_{1:T}\mid x) = p(z_1\mid x)\prod_{t=2}^T p(z_t\mid z_{1:t-1},x),
\end{equation*}
where
\begin{equation}
\label{eq:5}
  p(z_t\mid z_{1:t-1},x)  =  \int  g_x(z_t\mid s_t) \int f_x(s_t\mid s_{t-1})p(s_{t-1}\mid z_{1:t-1},x) ds_{t-1}  ds_t
\end{equation}
is the predictive likelihood.
 
In general it is impossible to evaluate the likelihood analytically, but it is often possible to approximate the likelihood using a particle filter \citep{Pitt2012,DoucetJohansen2011,Fearnhead2007}, by replacing $p(s_{t-1}\mid z_{1:t-1},x)$ in \eqref{eq:5} with a particle approximation
\begin{equation}
  \label{eq:7}
\hat{p}(ds_{t-1}\mid z_{1:t-1},x) = \sum_{i=1}^N w_{t-1}^{(i)}\delta_{s^{(i)}_{t-1}}(ds_{t-1}),  
\end{equation}
where $\delta_s$ is a Dirac mass at $s$ and $s_{t-1}^{(i)}$ is the $i\mbox{th}$ particle at $t-1$ with normalized weight $w_{t-1}^{(i)}$. An approximation to the likelihood \eqref{eq:5} is then given by the particle approximation $\{w_{t-1}^{(i)},s_{t-1}^{(i)}\}_{i=1}^N$, 
\[
\hat{p}(z_t\mid z_{1:t-1},x) = \sum_{i=1}^N \frac{\tilde{w}_t^{(i)}}{N}, 
\]
where $\tilde{w}_t^{(i)}$ is the $i\mbox{th}$ unnormalized importance weight at $t$. Using the Metropolis--Hastings algorithm \eqref{eq:13} we can target the exact posterior density $\pi(x)$ as outlined in Section \ref{sec:part-marg-metr}. Full details are given in the Supplementary Material.

The particle Langevin algorithm requires an estimate of the gradient of the log posterior density, $\nabla \log \pi(x) = \nabla \log p(z_{1:T}\mid x) + \nabla \log p(x)$. Assuming $\nabla \log p(x)$ is known, it is then only necessary to approximate the score vector $\nabla \log p(z_{1:T}\mid x)$ with a particle approximation of Fisher's identity \citep{cappe2005inference}
\begin{eqnarray}
\label{eq:fisher}
  \nabla \log p(z_{1:T}\mid x)
  &=& E \left\{\nabla \log p(S_{1:T},z_{1:T}\mid x)\mid z_{1:T},x \right\},
\end{eqnarray}
which is the expectation, with respect to $p(s_{1:T}|z_{1:T},x)$, of 
\begin{equation*}
  \nabla\log p(s_{1:T},z_{1:T}\mid x) =  \sum_{t=1}^T \nabla \log g_x(z_t\mid s_t) + \nabla \log f_x(s_t\mid s_{t-1})  
\end{equation*}
over the path $s_{1:T}$, where we have used the notation $f_x(s_1\mid s_0)=\mu_x(s_1)$.

A particle approximation is obtained by running the particle filter for $t=1,\ldots,T$ and storing the particle path $s_{1:T}^{(i)}$. Using the method of \cite{Poyiadjis2011}, the score vector is approximated by
\[
 \nabla \log \hat{p}(z_{1:T}\mid x)= \sum_{i=1}^N w_T^{(i)}\nabla\log p(s_{1:T}^{(i)},z_{1:T}\mid x),
\]
where $w_T^{(i)}$ is an importance weight.

With this approach the variance of the score estimate increases quadratically with $T$. 
\cite{Poyiadjis2011} suggest an alternative particle filter algorithm, which avoids the quadratically increasing variance, but at the expense of a computational cost that is quadratic in the number of particles. Instead, we use the algorithm of \cite{Nemeth2013}, which uses kernel density estimation and Rao--Blackwellization to substantially reduce the Monte Carlo variance, but still maintains an algorithm whose computational cost is linear in the number of particles; see the Supplementary Material. Importantly, the theory presented in Section \ref{sec:optim-scal-results} is not tied to any particular method for approximating the gradient of the log posterior density, and as such, alternative approaches proposed by \cite{Poyiadjis2011}, \cite{Ionides2011}, \cite{Dahlin2014} and others, are equally supported by our theoretical results.

\subsection{Linear Gaussian Model}
\label{sec:line-gauss-model}

This section provides simulation results to support the theory
outlined in Section \ref{sec:optim-scal-results}. We show that, while our theory is based on the limit as the number of
parameters tends to infinity, it adequately describes the
empirical results for a target with a finite number of parameters. 

We start by considering the following linear Gaussian state-space model, where it is possible to estimate the posterior density $p(x\mid z_{1:T})$, and its gradient, exactly with the Kalman filter \citep{durbin2001time},
\begin{equation*}
  z_t = \alpha + \beta s_t + \tau_\epsilon \nu_t, \quad s_t  = \mu + \phi s_{t-1} + \sigma_\epsilon \eta_t,  \quad s_0 \sim \mathcal{N}\{\mu/(1-\phi),\sigma_\epsilon^2/(1-\phi^2)\},
\end{equation*}
where $\nu_t$ and $\eta_t$ are standard independent Gaussian random
variables and the vector of model parameters is $x = (\alpha,\beta,\tau_\epsilon,\mu,\phi,\sigma_\epsilon)^{\mathrm{T}}$.
We simulated 500 observations from the model with parameters $x =
(0.2,1,1,0.1,0.9,0.15)^{\mathrm{T}}$, and defined the following prior distributions:
\[\left(\begin{matrix}   \alpha\\  \beta \end{matrix} \right) \sim \mathcal{N}\left\{\left(\begin{matrix} 0.3 \\1.2 \end{matrix} \right),\tau_\epsilon^2 \left(\begin{matrix}
  0.25  & 0 \\
  0  & 0.5 
 \end{matrix} \right)\right\}, \quad
\tau_\epsilon^2 \sim \mathrm{Inverse~Gamma}(1,7/20),
\]
$\mu \sim \mathcal{N}(0.15,0.5)$, $(\phi+1)/2 \sim \mathrm{Beta}(20,5)$ and $\sigma_\epsilon^2 \sim \mathrm{Inverse~Gamma}(2,1/40)$.

The parameters $(\phi,\sigma_\epsilon,\tau_\epsilon)$ are constrained as $|\phi|<1$, $\sigma_\epsilon >0$ and $\tau_\epsilon
>0$. These parameters are transformed as $\tanh\phi$,
$\log\sigma_\epsilon$ and $\log\tau_\epsilon$ to implement the
particle Langevin and random-walk proposals on the unconstrained
space.

For this model it is possible to use the fully adapted particle filter \citep{Pitt1999a} using the optimal proposal for the latent states, which, compared to the simpler bootstrap filter \citep{Gordon1993}, reduces the variance in the posterior estimates.  The particle Langevin algorithm was run for 100,000 iterations with $\lambda^2 = \gamma_i^2 \times 1.125^2/6^{-1/3} \times \hat{V}$, where $\gamma = (0.25,0.5,0.75,1,1.25,1.5,1.75,2)$ and $\hat{V}$ is the empirical posterior covariance estimated from a pilot run. Estimates of the posterior density and gradient of the log posterior density were calculated using a particle filter with particles $N = (200,100,70,40,20,5,10,1)$; see the Supplementary Material.

Figure \ref{fig:scaling_vs_sigma} shows the efficiency of the particle Langevin algorithm for various scalings $\gamma$ and noise $\sigma^2$. Dividing the minimum effective sample size, taken over the parameters, by the computational time of the algorithm provides a practical measure of efficiency corresponding to the theoretical measure in \eqref{eqn.eff.CPU}.

\begin{figure}[t!]
  \centering
  \includegraphics[scale=0.5]{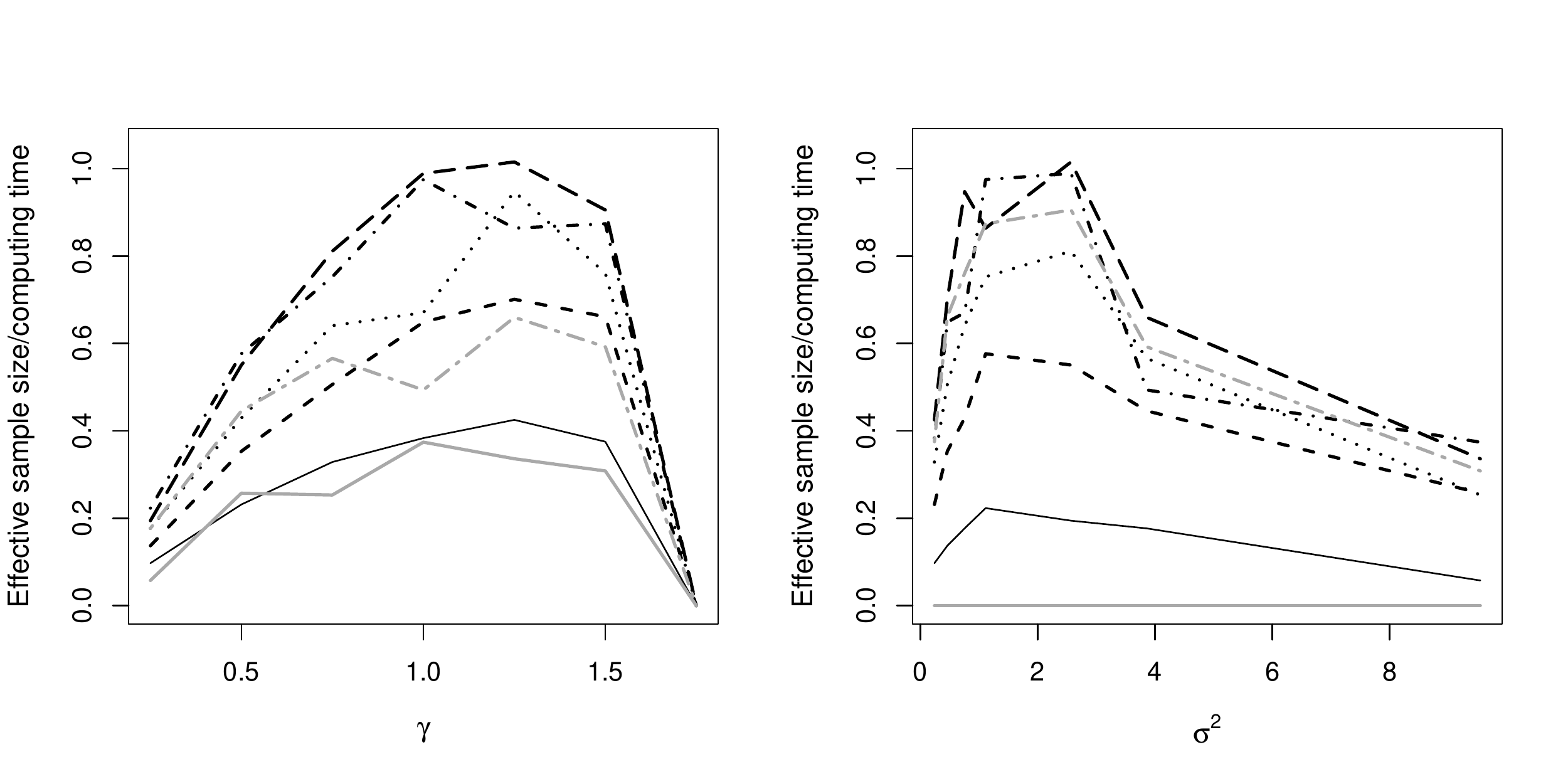}
  \caption{Empirical efficiency measured as the minimum effective sample size per computational second. The left panel gives the efficiency plotted against $\gamma$ for $(N=5~ \color{lightgray}\protect\rule[0.5ex]{0.6cm}{0.5pt}\color{black},N=10~\color{lightgray}\protect\rule[0.5ex]{0.1cm}{0.7pt}~\protect\rule[0.5ex]{0.4cm}{0.7pt}~\protect\rule[0.5ex]{0.1cm}{0.7pt}~\protect\rule[0.5ex]{0.4cm}{0.7pt} \color{black},N=20 ~\protect\rule[0.5ex]{0.2cm}{0.5pt}~\protect\rule[0.5ex]{0.2cm}{0.5pt}~\protect\rule[0.5ex]{0.2cm}{0.5pt} ,N=40~\boldsymbol{-} \cdot \boldsymbol{-} \cdot \boldsymbol{-},N=70~\cdot\cdot\cdot\cdot\cdot\cdot,N=100 ~\protect\rule[0.5ex]{0.1cm}{0.5pt}~\protect\rule[0.5ex]{0.1cm}{0.5pt}~\protect\rule[0.5ex]{0.1cm}{0.5pt}~\protect\rule[0.5ex]{0.1cm}{0.5pt}~\protect\rule[0.5ex]{0.1cm}{0.5pt},N=200~\protect\rule[0.5ex]{0.6cm}{0.5pt})$. The right panel gives the efficiency plotted against $\sigma^2$, estimated at the true parameter, for various scalings $(\gamma=0.25~\protect\rule[0.5ex]{0.6cm}{0.5pt},\gamma=0.5 ~\protect\rule[0.5ex]{0.1cm}{0.5pt}~\protect\rule[0.5ex]{0.1cm}{0.5pt}~\protect\rule[0.5ex]{0.1cm}{0.5pt}~\protect\rule[0.5ex]{0.1cm}{0.5pt}~\protect\rule[0.5ex]{0.1cm}{0.5pt},\gamma=0.75~\cdot\cdot\cdot\cdot\cdot\cdot,\gamma=1~\boldsymbol{-} \cdot \boldsymbol{-} \cdot \boldsymbol{-},\gamma=1.25 ~\protect\rule[0.5ex]{0.2cm}{0.5pt}~\protect\rule[0.5ex]{0.2cm}{0.5pt}~\protect\rule[0.5ex]{0.2cm}{0.5pt},\gamma=1.5~\color{lightgray}\protect\rule[0.5ex]{0.1cm}{0.7pt}~\protect\rule[0.5ex]{0.4cm}{0.7pt}~\protect\rule[0.5ex]{0.1cm}{0.7pt}~\protect\rule[0.5ex]{0.4cm}{0.7pt} \color{black},\gamma=1.75~ \color{lightgray}\protect\rule[0.5ex]{0.6cm}{0.5pt}\color{black})$.}
  \label{fig:scaling_vs_sigma}
\end{figure}

The left panel of Fig. \ref{fig:scaling_vs_sigma} shows that, initially,
increasing the number of particles leads to a more efficient sampler. However, beyond $20$ particles, the increase in computational 
cost outweighs the further improvement in mixing. Setting $N=20$ results in a noisy estimate of the posterior density
with $\sigma^2 \approx 2.6$, supporting Corollary
\ref{sec:opt-acc-rate}; the optimal acceptance rate was $19\%$, slightly above the theoretical optimum. Also, the insensitivity of the optimal scaling to the noise
variance, as shown in Figure \ref{fig:contours}, is seen here as the efficiency is maximized for $\gamma$ 
between $1$ and $1.5$ regardless of the number of particles; similarly
the right panel shows the same insensitivity of the optimal
variance to the scaling, with efficiency maximized for
$\sigma^2$ between $1.5$ and $3$, regardless of the scaling. Both of
these insensitivities are predicted by the theory established in Section \ref{sec:optim-scal-results}.

\subsection{Mixture model of autoregressive experts}
\label{sec:mixt-autor-experts}

We now use a real data example from \cite{Pitt2012} to illustrate the improvement of using the particle Langevin algorithm \eqref{eq:pmala} over the particle random-walk algorithm. Moreover, we show that estimating the gradient using the $\mathcal{O}(N)$ algorithm of \cite{Nemeth2013} is more efficient than the $\mathcal{O}(N^2)$ algorithm of \cite{Poyiadjis2011}.

This example uses a two-component mixture of experts model observed with noise. Each of the experts is represented by a first order autoregressive process, where the mixing of the experts is probabilistic rather than deterministic. The model is defined as
\begin{eqnarray}
  \label{eq:mixExperts}
  z_t = s_t + \tau_\epsilon \nu_t, \quad\quad  s_t = \psi_{J_t}+\phi_{J_t}s_{t-1}+\sigma_{J_t} \eta_t \quad J_t=1,2\\   \nonumber
  \mathrm{pr}(J_t=1\mid s_{t-1},s_{t-2}) = \frac{\exp\{\xi_1+\xi_2 s_{t-1}+\xi_3 (s_{t-1}-s_{t-2})\}}{1+\exp\{\xi_1+\xi_2 s_{t-1}+\xi_3 (s_{t-1}-s_{t-2})\}},
\end{eqnarray}
where $\nu_t$ and $\eta_t$ are standard independent Gaussian random variables and there are 10 model parameters $x=(\tau_\epsilon,\psi_1,\psi_2,\phi_1,\phi_2,\sigma_1,\sigma_2,\xi_1,\xi_2,\xi_3)^{\mathrm{T}}$.

\cite{Pitt2012} used the mixture of autoregressive experts to model the growth of US gross domestic product from the second quarter of 1984 to the third quarter of 2010. This model follows previous observations that economic cycles display nonlinear and non-Gaussian features \citep{hamilton89}. Including measurement noise in the model accounts for adjustments made to the data between the first and final release \citep{zellner92}. We impose the constraint $\psi_1(1-\phi_1)<\psi_2(1-\phi_2)$ to ensure that the mean of expert one is less than that of expert two. This implies that the first expert is identified as a low growth regime.

A particle filter approach to this problem is ideal if we assume measurement error in the data. Standard Markov chain Monte Carlo methods could be applied on this model where the latent states are sampled conditional on the parameters and vice-versa \citep{pittbook2010}. However, this would cause the sampler to mix slowly and would ultimately be less efficient than a particle filter implementation, whereby the latent states are integrated out. We compare the particle Langevin algorithm against the particle random-walk algorithm, as implemented in \cite{Pitt2012}. For both methods we implement a fully adapted particle filter, where the number of particles were tuned to give a variance of less than 3 for the log posterior density.

We ran the particle Markov chain Monte Carlo algorithm for 100,000 iterations, discarding the first half as burn-in. We compare the random-walk proposal, scaled as $\lambda^2 = 2.526^2/10 \times \hat{V}$, against the particle Langevin proposal, where $\lambda^2 = 1.125^2/10^{-1/3} \times \hat{V}$ and $\hat{V}$ is an estimate of the posterior covariance taken from a pilot run. A Gaussian prior density is assumed for $x$, where constrained parameters are transformed appropriately, and the hyper-parameters are given in \cite{Pitt2012}. Table \ref{tab:ess} gives a comparison of the proposals, including a particle Langevin algorithm using the $\mathcal{O}(N^2)$ gradient estimator of \cite{Poyiadjis2011}. The minimum and maximum effective sample size per computational minute, taken over 10 simulations, are reported.

\begin{table}[t!]
\caption{Empirical effective sample size per computational time}
\begin{tabular}{c|c|cccccccccc}
\multirow{2}{*}{Algorithm} & \multicolumn{11}{c}{Effective sample size per computation time} \\
 & & $\tau_\epsilon$ & $\psi_1$&  $\psi_2$& $\phi_1$ & $\phi_2$ & $\sigma_1$ & $\sigma_2$ & $\xi_1$ & $\xi_2$ & $\xi_3$ \\
\hline
\multirow{2}{*}{Particle random walk} & Min & 3.39 &2.96& 1.65& 2.15& 1.96& 1.38 &2.16 &2.54 &2.05 &2.09 \\
 & Max &   4.65 &3.68 &3.15 &3.68 &3.48 &2.82& 3.56& 4.20 &3.32& 3.71 \\
\multirow{2}{*}{Particle Langevin $\mathcal{O}(N)$} & Min &   4.11 &3.21 &4.77 &3.57 &4.18 &2.60 &3.68 &4.59 &3.32 &3.08 \\
 & Max &  5.12 &5.71 &6.37 &6.12 &6.43 &5.47 &6.22 &7.34 &7.02 &6.10 \\
\multirow{2}{*}{Poyiadjis $\mathcal{O}(N^2)$} & Min & 0.76 &0.60 &1.00 &0.96& 0.47& 0.33& 0.90 &1.06 &0.59 &0.59  \\
 & Max &  1.25 &1.19 &1.37 &1.35 &1.35 &1.17 &1.26 &1.82 &1.23& 0.96   \\
\end{tabular}
\label{tab:ess}
\end{table}

The results from the simulation study are summarized in Table \ref{tab:ess}. There is a significant improvement in terms of effective sample size when using the particle Langevin proposal compared to the random-walk proposal. The effective sample size of the \cite{Poyiadjis2011} $\mathcal{O}(N^2)$ algorithm is approximately equal to that of our particle Langevin algorithm, but when taking into account the computational cost, this proposal performs worse than the random-walk algorithm. Therefore, it is important to estimate the gradient of the log posterior density with the same computational cost used to estimate the log posterior density in order for the gradient information to be beneficial.

\section{Discussion}
\label{sec:discussion}

Our theory identifies three
distinct asymptotic regimes corresponding to three different ranges
of control over errors in the estimate of the gradient of the
log posterior density. We have shown that, if there is no control of these errors, then the particle Langevin algorithm is asymptotically no more
efficient than the particle random-walk Metropolis algorithm. By contrast, if there is
sufficient control, the particle
Langevin algorithm attains the same asymptotic advantage in
efficiency over the particle random-walk algorithm as the Metropolis
adjusted Langevin algorithm enjoys over its random-walk Metropolis counterpart.

In the preferred regime, and specifically when the estimate of the
log posterior density is generated by a particle filter, we identify an
optimal variance for the error in the log posterior density of approximately
$3.0$ and an optimal
acceptance rate of approximately $15\%$.  We also find that the optimal scaling is insensitive
to the choice of variance and vice-versa. In general, however, the regime is not known, and so,
 conditional on a fixed but arbitrary
number of particles, we provide a mechanism for tuning the scaling of
the proposal by aiming for a maximin acceptance rate that is
robust to the  regime. This ensures that the resulting algorithm will
achieve at least $90\%$ of the efficiency that it would were the regime known and
the best scaling for that regime chosen.

Our results are the latest in a number of results concerning at optimal
implementations of pseudo-marginal and particle Markov chain Monte Carlo algorithms
\cite[]{Pitt2012,Sherlock:2015,Doucet:2015}. Using similar techniques
to those in this article, \cite{Sherlock:2015} 
identified an optimal variance of $3.3$ for the target density in the particle random-walk Metropolis algorithm, and found that the optimal variance is insensitive to the scaling and
vice-versa. \cite{Doucet:2015} analyzed non-asymptotic bounds on the integrated
autocorrelation time and suggested that for any Metropolis--Hastings
algorithm, the optimal variance should be between $0.85$ and
$2.82$, also suggesting an insensitivity. In high dimensions,
one should tune to the variance suggested by the asymptotic theory, but 
empirical studies on both the particle random-walk and the particle Langevin algorithms have shown that in low dimensions the optimal variance
is typically less than $3$.
Given that all our assumptions hold at best approximately in practice,
we would recommend two possible tuning strategies.
The first strategy is to evaluate the variance at several points in the main posterior mass
and ensure that the largest of these is slightly lower than $3$;
this is justified both because of the above mentioned empirical
findings and because we expect the optimal variance to be
smaller in regimes (1) and (2). Then
tune the scaling to achieve the acceptance rate given by Figure
\ref{fig:minimax}. For the second strategy,  start with a sensible scaling, find the number of
particles that optimizes the overall efficiency, for example effective sample size per second, then with
this number of particles, find
the scaling which optimizes efficiency, for example effective sample size.

\section*{Acknowledgements}
\label{sec:acknowledgements}

The authors are grateful to the editor, associate editor and referees whose comments and advice have greatly improved this paper. We thank Jeff Rosenthal for providing the Mathematica scripts from \cite{Roberts1998}. This research was supported by the EPSRC i-like grant and the STOR-i Centre for Doctoral Training.

\bibliography{library,library1}
\bibliographystyle{apalike}

\pagebreak
\begin{center}
\textbf{\large Supplemental Materials}
\end{center}
\setcounter{equation}{0}
\setcounter{section}{0}
\setcounter{figure}{0}
\setcounter{table}{0}
\setcounter{page}{1}
\makeatletter
\renewcommand{\theequation}{S\arabic{equation}}
\renewcommand{\thefigure}{S\arabic{figure}}
\renewcommand{\bibnumfmt}[1]{[S#1]}
\renewcommand{\citenumfont}[1]{S#1}

\section{Proof of Theorem \ref{thrm.main}} 
\label{sec:proof-theorem-1}

The proposal density for any given component is
\[
q(x,y)=\left(\lambda_n^2 2\pi\right)^{-1/2}
\exp\left\{-\frac{1}{2\lambda_n^2}\left(y-x-\frac{1}{2}\lambda_n^2g'(x)\right)^2\right\}.
\]
Define
\[
R(x_i^n,Y^n_i)=\log\left\{\frac{f(Y^n_i)q(Y^n_i,x_i^n)}{f(x_i^n)q((x_i^n,Y^n_i))}\right\}
\]
and
$T_n(X^n,Y^n)=\sum_{i=1}^n R(X^n_i,Y^n_i)$, 
so that the acceptance probability is
\begin{equation}
\label{eqn.acc.prop.n}
\alpha_n(x^n,w^n;Y^n,V^n)=1\wedge \exp\left\{V^n-w^n+T_n(x^n,Y^n)\right\},
\end{equation}
where $V^n$ and $W^n$ are given in \eqref{eqn.SARa} and \eqref{eqn.targ.W}. Note, we use the notation $a \wedge b = \mbox{min} (a,b)$.
\begin{proposition}
\label{prop.mathematica}
\begin{eqnarray}
\nonumber
R(x_i^n,Y^n_i)
&=&C_3(x_i^n,Z_i)\lambda_n^3+C_4(x_i^n,Z_i)\lambda_n^4+C_5(x_i^n,Z_i)\lambda_n^5\\
&&+C_6(x_i^n,Z_i)\lambda_n^6+C_7(x_i^n,Z_i,\lambda_n),
\label{eqn.ri.def}
\end{eqnarray}
where 
\[
C_3(x_i^n,Z_i)=-\frac{1}{12}\ell^3\left\{3Z_ig'(x_i^n)g''(x_i^n)+Z_i^3g'''(x_i^n)\right\},
\]
and where $C_4(x_i^n,Z_i)$, $C_5(x_i^n,Z_i)$ and $C_6(x_i^n,Z_i)$ are also polynomials in $Z_i$ and the derivatives of $g$. Furthermore, if $E_Z$ denotes expectation with $Z\sim \mathcal{N}(0,1)$ and $E_X$ denotes expectation with $X$ having the density $f(\cdot)$, then
\begin{equation}
\label{eqn.expect.zero}
E_X[\Expects{Z}{C_3(X,Z)}]
=
E_X[\Expects{Z}{C_4(X,Z)}]
=
E_X[\Expects{Z}{C_5(X,Z)}]=0,
\end{equation}
whereas
\begin{equation}
\label{eqn.expect.nonzero}
E_X[\Expects{Z}{C_3(X,Z)^2}]=\ell^6K^2=-2E_X[\Expects{Z}{C_6(X,Z)}]>0.
\end{equation}
Also 
\begin{equation}
\label{eqn.finite.vars}
\Var{C_4(X,Z)}<\infty,~\Var{C_5(X,Z)}<\infty, \Var{C_6(X,Z)}<\infty,
\end{equation}
where $\mathrm{var}$ denotes variance over both $Z$ and $X$.
Finally
\begin{equation}
\label{eqn.remainder}
\Expects{Z}{\Abs{C_7(x_i^n,Z_i,\lambda_n)}}\le n^{-7/6}p(x_i^n),
\end{equation}
where $p(x)$ is a polynomial in $x$.
\end{proposition}

\begin{proof}
As in \cite{Roberts1998}, equation \eqref{eqn.ri.def} follows by
Taylor expansion of $g$ and its derivatives using MATHEMATICA
\citep{Wolfram} and collecting terms in powers of $n$. Straightforward
inspection shows that $C_3$ has the claimed form and that $C_4,~C_5$
and $C_6$ are also polynomials in $Z$ and the derivatives of $g$, as
claimed. All terms in both  $C_3$ and $C_5$ contain odd powers of $Z$
and so their expectations are zero. Equation
\eqref{eqn.expect.nonzero}, and the fact that the expectation of $C_4$
is zero, follows after integrating by parts where expectations of
products of the derivatives of $g$ are being taken with respect to the
density $e^{g(x)}$. Thus the equivalent form of $K$ defined in
\eqref{eq:K} is real and positive. Equation \eqref{eqn.finite.vars} follows from the polynomial form for $C_4,~C_5$ and $C_6$ and assumptions \eqref{eqn.poly.mom.g} and \eqref{eqn.finite.moments}.

Using the remainder formula of the Taylor series expansion we may derive the bound
\[
\Abs{C_7(x_i^n,Z_i,\lambda_n)}\le n^{-7/6}p_*(x_i^n,w_i),
\]
for some polynomial $p_*$, with $\Abs{w_i}\le \Abs{Z_i}$. But for any
polynomial $p_*(x,w)\le A(1+x^N)(1+w^N)$, with a sufficiently large $A$
and for a sufficiently large even integer $N$, \eqref{eqn.remainder} follows with
$p(x)=A\Expect{(1+Z^N)}(1+x^N)$. 
\end{proof}

Proposition \ref{prop.mathematica} allows us to find the limiting distribution of one of the key terms in the acceptance probability of the algorithm when the Markov chain on $(X,W)$ is stationary, \eqref{eqn.acc.prop.n}.

\begin{lemma}
\label{lemma.CLT}
\[
T_n(X^n,Y^n)\Rightarrow T\sim \mathcal{N}\left(-\frac{1}{2}\ell^6K^2,\ell^6K^2\right).
\]
\end{lemma}
\begin{proof}
First note that, by \eqref{eqn.remainder},
\[
\Expect{\Abs{\sum_{i=1}^nC_7(X_i^n,Z_i,\lambda_n)}}
\le
\sum_{i=1}^n\Expect{\Abs{C_7(X_i^n,Z_i,\lambda_n)}}
\le n^{-1/6}\Expect{p(X)}.
\]
However $\Expect{p(X)}<\infty$ by \eqref{eqn.finite.moments} so, by Markov's inequality, $\sum_{i=1}^nC_7(X_i^n,Z_i,\lambda_n) \rightarrow 0$ in probability as $n \rightarrow \infty$. By Slutsky's Theorem it is therefore sufficient to show that in probability $T'_n \rightarrow T\sim \mathcal{N}\left(-\frac{1}{2}\ell^6K^2,\ell^6K^2\right)$, where we define
\[
T_n'=\sum_{i=1}^n\left\{C_3(X_i^n,Z_i)n^{-1/2}+C_4(X_i^n,Z_i)n^{-2/3}+C_5(X_i^n,Z_i)n^{-5/6}\\
+C_6(X_i^n,Z_i)n^{-1}\right\}.
\]
Combining \eqref{eqn.expect.zero}, \eqref{eqn.expect.nonzero} and \eqref{eqn.finite.vars}
\begin{eqnarray*}
E(T_n')&=&-\frac{1}{2}\ell^6K^2,\\
\mathrm{var}(T_n')&=&\Var{C_3(X,Z)}+O(n^{-1/6})\rightarrow\ell^6K^2
\end{eqnarray*}
in probability.
Moreover $T_n'$ is the sum of $n$ independent and identically
distributed terms, so the result follows by the central limit
theorem. 
\end{proof}

Thus 
\begin{equation}
T_n+V^n-W^n\Rightarrow \mathcal{N}\left(B-\frac{1}{2}\ell^6K^2,\ell^6K^2\right).
\end{equation}
Now  if $U\sim \mathcal{N}(a,b^2)$ then $E(1 \wedge e^{U})=\Phi(a/b)+e^{a+b^2/2}\Phi(-b-a/b)$ (e.g., \cite{RobertsGelmanGilks1997}). Since $\alpha_n=E(1 \wedge e^{T_n+B^n})$, we may apply the Bounded Convergence Theorem to see that
\[
\lim_{n\rightarrow \infty}\alpha_n= E(1 \wedge e^{T+B})
= 2 E \left\{
\Phi\left(\frac{B}{\ell^3K}-\frac{1}{2}\ell^3K\right)\right\}.
\]
proving the first part of Theorem \ref{thrm.main}. 

To prove the second result, we first note the following 
\begin{proposition}
\label{prop.SqDist}
\begin{eqnarray*}
\lim_{n\rightarrow\infty}E\left(n^{-2/3}\Norm{Y^n-X^n}^2\right)&=&\ell^2\\
\lim_{n\rightarrow\infty}\Expect{\left(n^{-2/3}\Norm{Y^n-X^n}^2-\ell^2\right)^2}&=&0.
\end{eqnarray*}
\end{proposition}
\begin{proof} To simplify the exposition we suppress the superscripts, $n$, in $X^n$ and $Y^n$. Firstly,
\begin{eqnarray*}
n^{-2/3}E\left(\Norm{Y-X}^2\right)&=&n^{1/3}\Expect{(Y_1-X_1)^2}\\
&=&
n^{1/3}E\left[\left\{\ell n^{-1/6}Z_1+\frac{1}{2}\ell^2n^{-1/3}g'(X_1)\right\}^2\right]\\
&=&\ell^2+\frac{1}{4}\ell^4n^{-1/3}\Expect{g'(X_1)^2}.
\end{eqnarray*}
By assumptions \eqref{eqn.poly.mom.g} and \eqref{eqn.finite.moments},
$\Expect{g'(X_1)^2}<\infty$ and the first result follows. Also as
$n\rightarrow \infty$,
\begin{eqnarray*}
\mathrm{var}\left(n^{-2/3} \Norm{Y-X}^2-\ell^2\right)&=&n^{-4/3}\Var{\sum_{i=1}^n(Y_1-X_1)^2}\\
&=&n^{-1/3}\mathrm{var}\left[\left\{\ell
    n^{-1/6}Z_1+\frac{1}{2}\ell^2n^{-1/3}g'(X_1)\right\}^2\right]
\rightarrow 0,
\end{eqnarray*}
by \eqref{eqn.poly.mom.g} and \eqref{eqn.finite.moments}. This,
combined with the first part of this proposition proves the second
part. 
\end{proof}

To complete the proof, we abbreviate $\alpha_n(X^n,W^n;Y^n,V^n)$ to $A_n$.
Now
\[
\Abs{E\left(n^{-2/3}\Norm{Y^n-X^n}^2A_n\right)-\ell^2\alphabar}
\le \Abs{\Expect{\left(n^{-2/3}\Norm{Y^n-X^n}^2-\ell^2\right)A_n}}
+\Abs{\Expect{\ell^2\left(A_n-\alphabar\right)}}.
\]
The second term on the right hand side converges to zero by the first part of Theorem \ref{thrm.main}. The Cauchy--Schwarz inequality bounds the first term on the right hand side by
\[
\Expect{\left(n^{-2/3}\Norm{Y^n-X^n}^2-\ell^2\right)^2}^{1/2}E(A_n^2)^{1/2}.
\]
The first term converges to zero by Proposition \ref{prop.SqDist} and the second term is bounded.

\section{Proof of Corollary \ref{sec:opt-acc-rate}}
\label{sec:proof-corollary1}

First note that for some $Z\sim N(0,1)$ that is independent of $B$,
\[
\Phi\left(\frac{B}{\ell^3K}-\frac{\ell^3K}{2}\right)
=
\mathbb{P}\left(\ell^3KZ-B\le -\frac{1}{2}\ell^6K^2\right)
=\Phi\left\{-\left(\frac{1}{2}\ell^6K^2+\sigma^2\right)\left(\ell^6K^2+2\sigma^2\right)^{-1/2}\right\}.
\]
So
\[
\alpha(\ell,\sigma^2)=2\Phi\left\{-\frac{1}{2}(\ell^6K^2+2\sigma^2)^{-1/2}\right\}.
\]
Set $a^2=K^2\ell^6$ and $b^2=2\sigma^2$ then
\[
\mbox{Eff}(\ell,\sigma^2)\propto a^{2/3}b^2\Phi\left\{-\frac{1}{2}(a^2+b^2)^{-1/2}\right\}.
\]
Given $a^2+b^2$, $a^{2/3}b^2$ is maximized when $b^2=3a^2$, at which
point the efficiency is proportional to
$a^{8/3}\Phi\left(-a\right)$. Numerical optimization shows that this
function is maximized at $\hat{a}\approx 1.423$, and thus the optimal
acceptance rate is $\hat{\alpha}=2\Phi(-\hat{a})\approx 15.47\%$. As a
result, the optimal scaling and variance are $\ell_{\mathrm{opt}}\approx 1.125K^{-1/3}$ and $\sigma^2_{\mathrm{opt}} \approx 3.038$, as given in the statement.

\section{Proof of Theorem \ref{sec:scaling-grads}}
\label{sec:proof-theorem2}

For the sake of brevity, we shall prove statements (1), (2) and (3)
of Theorem \ref{sec:scaling-grads} together rather than
separately. Throughout the proof, therefore, the superscript $*$ will
be used to denote a superscript that could be replaced by $(1)$, $(2)$ or $(3)$ according to the case in the statement of Theorem \ref{sec:scaling-grads} that is being considered. 

For $*\in\{(1),(2),(3)\}$ let $R^*(x_i^n,Y^n_i)$ be the log Metropolis--Hastings ratio where the proposal, $q^*(x_i^n,Y_i^n)$, is the particle Langevin proposal given in
\eqref{eqn.define.Y.bias}, and let 
\begin{equation}
\label{eqn.define.Tstar}
T^*_n(X^n,Y^n)=\sum_{i=1}^n
R^*(X_i^n,Y_i^n).
\end{equation} 
We also define $U_i=(U_{x_i},U_{y_i})$, to be the vector of (zero mean and unit variance) noise terms in the $i\mbox{th}$ component of the gradient estimate used, respectively, in the particle Langevin proposal from the current value and the proposal for the corresponding reverse move from the proposed value.

The proof commences with an analogous result to Proposition
\ref{prop.mathematica} from Section \ref{sec:proof-theorem-1}.

\begin{proposition}
\label{prop.mathematica.two}
Let $R(x_i^n,Y_i)$ 
be the idealized particle Langevin algorithm term from
Proposition \ref{prop.mathematica}. Then for (*) in (1), (2) or (3)
\begin{eqnarray}
\nonumber
R^*(x_i^n,Y_i^n)&=&R(x_i^n,Y_i^n)+C_{1,1}(x_i^n,U_i,Z_i)\lambda_nn^{-\kappa}+C_{2,1}(x_i^n,U_i,Z_i)\lambda_n^2n^{-\kappa}+C_{3,1}(x_i^n,U_i,Z_i)\lambda_n^3n^{-\kappa}\\
&&+C_{4,1}(x_i^n,U_i,Z_i)\lambda_n^4n^{-\kappa}+C_{2,2}(x_i^n,U_i,Z_i)\lambda_n^2n^{-2\kappa}+C_{r}(x_i^n,U_i,Z_i,n),
\label{eqn.Rstar.expand}
\end{eqnarray}
Here 
\begin{eqnarray*}
C_{1,1}=- \frac{1}{2}\ell\tau U_{x_i^n}Z_i - \frac{1}{2}\ell\tau U_{y_i^n}Z_i - b(x_i^n)\ell Z_i,
\end{eqnarray*}
and $C_{2,1}C_{3,1},C_{4,1}$ and $C_{2,2}$ are all polynomials in $Z$,
in derivatives of $g(x)$, and in $b(x)$ and its derivatives.
Furthermore, let $E_{X,U,Z}$ denote expectation
with respect to $X$ having the density $f(\cdot)$,  $Z\sim \mathcal{N}(0,1)$, and with respect to $U_x$ and $U_y$ with
$E(U_x)=E(U_y)=0$ and $\mathrm{var}(U_x)=\mathrm{var}(U_y)=1$. Then
\begin{eqnarray}
{\Expects{X,U,Z}{C_{1,1}(X,U,Z)}}
&=&
{\Expects{X,U,Z}{C_{2,1}(X,U,Z)}}
=
{\Expects{X,U,Z}{C_{3,1}(X,U,Z)}}=0
\label{eqn.expect.zero,bias}
\end{eqnarray}
\begin{eqnarray}
  \label{eqn.equal.Kstar}
{\Expects{X,U,Z}{C_{1,1}(X,U,Z)^2}} &=&
-2{\Expects{X,U,Z}{C_{2,2}(X,U,Z)}}=K_*^2\\
  \label{eqn.equal.Kstarstar}
{\Expects{X,U,Z}{C_{3}(X,U,Z)C_{1,1}(X,U,Z)}} &=&
-{\Expects{X,U,Z}{C_{4,1}(X,U,Z)}}=K_{**}.
\end{eqnarray}
where $K_*$ and $K_{**}$ are defined in \eqref{eqn.defn.Kstar} and
\eqref{eqn.defn.Kstarstar}. Also 
\begin{eqnarray}
\nonumber
\Var{C_{1,1}(X,U,Z)}&<&\infty,~\Var{C_{2,1}(X,U,Z)}<\infty,\Var{C_{3,1}(X,U,Z)}<\infty,\\
\Var{C_{4,1}(X,U,Z)}&<&\infty, \Var{C_{2,2}(X,U,Z)}<\infty,
\label{eqn.finite.vars.two}
\end{eqnarray}
where $\mathrm{var}$ denotes variance over $Z$, $U$ and $X$.
Finally
\begin{equation}
\label{eqn.remainder.two}
\Expects{U,Z}{\Abs{C_7(x_i^n,U_i,Z_i,\lambda_n)}}\le n^{-7/6}p(x_i^n),
\end{equation}
where $p(x)$ is a polynomial in $x$.
\end{proposition}

\begin{proof}
Writing $A(x)=b(x) + \tau U_{x}$ and $A(y)=b(y) + \tau
U_{y}$, after some algebra we obtain
\[
\log q^*(y_i^n,x_i^n) - \log q^*(x_i^n,y_i^n) =
\frac{1}{2}Z_i^2-\frac{1}{2}\left(Z_i+\frac{\lambda_n}{2}\left[g'(x_i^n)+g'(y_i^n)+n^{-\kappa}\{A(x_i^n)+A(y_i^n)\}\right]\right)^2.
\]
This, together with a simpler calculation for $\log \pi(y_i^n)-\log \pi(x_i^n)$,
shows that in a Taylor expansion of $R^*(x_i^n,y_i^n)$ about $x_i^n$, terms in
$n^{-a\kappa}~(a=1,\dots)$ must also be multiplied by $\lambda_n^b$
with $b\in\{a,a+1,\dots\}$. Consideration of the maximum possible size of all terms
in the Taylor expansion for the three different cases shows that it
must be of the form given in \eqref{eqn.Rstar.expand} with the
largest part of the remainder term being at most $O(n^{-7/6})$.

As with Proposition \ref{prop.mathematica}, the polynomial forms for
$C_{1,1}, C_{2,1}, C_{3,1}, C_{4,1}$ and $C_{2,2}$ are produced using MATHEMATICA \citep{Wolfram}, but this time by also Taylor expanding the term $b(y)$ in $y-x$. 

Clearly $E(C_{1,1})=0$ as the terms are multiples of $U_x$,
$U_y$ and odd powers of $Z$. The same argument can be used for the
expectations of $C_{3,1}$; however for $C_{2,1}$, $C_{4,1}$ and for
the relationships in \eqref{eqn.equal.Kstar} and
\eqref{eqn.equal.Kstarstar}, it must be used in tandem with
integration by parts with respect to the target density $e^{g(x)}$
and using assumptions \eqref{eqn.poly.mom.g},
\eqref{eqn.finite.moments} and \eqref{eq:cond.b}. 

We illustrate this by providing the form for $C_{2,1}$:
\[
C_{2,1}(X,U,Z)=-\frac{1}{2}\ell^2Z^2\left\{b'(X)+b(X)g'(X)+\tau U_{y}g'(X)\right\}.
\]
However \eqref{eqn.poly.mom.g} and \eqref{eq:cond.b} imply that 
$\int b'(x)e^{g(x)} ~d x=-\int b(x)g'(x)e^{g(x)}$.

The final two parts of the proposition follow from analogous arguments
to those used in Proposition \ref{prop.mathematica} provided
\eqref{eqn.mom.U} holds. 
\end{proof}

Integration by parts, the Cauchy--Schwarz inequality and then further
integration by parts gives
\begin{eqnarray*}
K^2_{**}&=&\frac{1}{16} E \left[b(X)\{g'(X)g''(X) + g'''(X)\}\right]^2
\le 
\frac{1}{16} E \left\{b(X)^2\right\} E \left[\{g'(X)g''(X)+g''(X)\}^2\right] \\
&=&\frac{1}{48} E \left\{b(X)^2\right\} E \left[3\{g'''(X)\}^2-3\{g''(X)\}^3\right] \le K_*^2K^2,
\end{eqnarray*}
so that $\ell^4K^2+2\ell^2K_{**}+K_*^2\ge 0$.


\begin{lemma}
\label{lemma.CLT.two}
For $*\in\{(1),(2),(3)\}$
\[
T^*_n(X^n,Y^n)\Rightarrow T^*\sim \mathcal{N}\left(-\frac{1}{2}a^{*},a^*\right),
\]
where $T^*$ is defined in \eqref{eqn.define.Tstar} and 
\[
a^{(1)} =\ell^2 K_*^2,
~a^{(2)} =\ell^2 K_*^2 + 2 \ell^4 K_{**}+\ell^6K^2,
~a^{(3)}=\ell^6K^2.
\]
\end{lemma}
\begin{proof}
As proved in Lemma \ref{lemma.CLT}, by Markov's inequality,
$\sum_{i=1}^nC_r(X_i^n,U_i,Z_i,n)\rightarrow 0$ in probability and therefore it is
sufficient to show that in probability
$T^*_n\rightarrow T \sim\mathcal{N}\left(-\frac{1}{2}a^{*},a^*\right)$, where
$*\in\{(1),(2),(3)\}$ and
$T^*_n=\sum_{i=1}^n\left\{R^*(x_i^n,Y_i^n,Z_i^n)-C_r(x_i^n,U_i,Z_i,n)\right\}$.
Table \ref{tab:coeff} shows the coefficient of each non-remainder term in
\eqref{eqn.Rstar.expand} in each of the three cases.

\begin{table}[t!]
\caption{Coefficients of \eqref{eqn.Rstar.expand} terms}
\begin{tabular}{c|lllllllll}
&$C_3$&$C_4$&$C_5$&$C_6$&$C_{1,1}$&$C_{2,1}$&$C_{3,1}$&$C_{4,1}$&$C_{2,2}$\\
\hline
(1)&$n^{-\frac{1}{2}-3\epsilon}$&$n^{-\frac{2}{3}-4\epsilon}$&$n^{-\frac{5}{6}-5\epsilon}$
&$n^{-1-6\epsilon}$& $n^{-\frac{1}{2}}$&$n^{-\frac{2}{3}-\epsilon}$&$n^{-\frac{5}{6}-2\epsilon}$&$n^{-1-3\epsilon}$&$n^{-1}$\\
(2)&$n^{-\frac{1}{2}}$&$n^{-\frac{2}{3}}$&$n^{-\frac{5}{6}}$
&$n^{-1}$& $n^{-\frac{1}{2}}$&$n^{-\frac{2}{3}}$&$n^{-\frac{5}{6}}$&$n^{-1}$&$n^{-1}$\\
(3)&$n^{-\frac{1}{2}}$&$n^{-\frac{2}{3}}$&$n^{-\frac{5}{6}}$
&$n^{-1}$& $n^{-\frac{1}{2}-\epsilon}$&$n^{-\frac{2}{3}-\epsilon}$&$n^{-\frac{5}{6}-\epsilon}$&$n^{-1-\epsilon}$&$n^{-1-2\epsilon}$\\
\end{tabular}
\label{tab:coeff}
\end{table}

Since $\epsilon>0$, as $n\rightarrow \infty$,
combining \eqref{eqn.expect.zero}, \eqref{eqn.expect.nonzero},
\eqref{eqn.expect.zero,bias}, \eqref{eqn.equal.Kstar} and
\eqref{eqn.equal.Kstarstar}  gives
\begin{eqnarray*}
E(T_n^{(1)})&=&-\frac{1}{2}n^{-6\epsilon}\ell^6K-n^{-3\epsilon}\ell^4K_{**}-\frac{1}{2}\ell^2K_*\rightarrow
-\frac{1}{2}\ell^2K_*,\\
E(T_n^{(2)})&=&-\frac{1}{2}\ell^6K-\ell^4K_{**}-\frac{1}{2}\ell^2K_*,\\
E(T_n^{(3)})&=&-\frac{1}{2}\ell^6K-n^{-\epsilon}\ell^4K_{**}-\frac{1}{2}n^{-2\epsilon}\ell^2K_*\rightarrow-\frac{1}{2}\ell^6K,
\end{eqnarray*}
in probability. Similarly, using \eqref{eqn.finite.vars} and \eqref{eqn.finite.vars.two},
\begin{eqnarray*}
\mathrm{var}(T_n^{(1)})&\rightarrow&E(C_{1,1}^2)=\ell^2K_*^2\\
\mathrm{var}(T_n^{(2)})&\rightarrow&\Expect{\left(C_3+C_{1,1}\right)^2}
=\ell^2K_*^2+2\ell^4K_{**}+\ell^6K^2\\
\mathrm{var}(T_n^{(3)})&\rightarrow&E(C_3^2)=\ell^6K^2,
\end{eqnarray*}
in probability.  Moreover, $T_n^{*}$ is the sum of $n$ independent and identically distributed terms, so the result follows by the central limit theorem. 
\end{proof}

The proof for the asymptotic acceptance rate,
$\alphabar_n(\ell,\sigma^2)\rightarrow \alpha^{*}(\ell,\sigma^2)$, is
completed using Lemma \ref{lemma.CLT.two} as in the proof of Theorem \ref{thrm.main}  by accounting for the distribution of the noise of the log-target.

Finally, as in the proof of Theorem \ref{thrm.main}, the limit for
the squared jump distance, $J_n(\ell,\sigma^2)$, follows from Proposition \ref{prop.SqDist}.

\section{Implementation details for the particle Langevin algorithm}
\label{sec:impl-deta-part}

Particle filters, also known as sequential Monte Carlo algorithms, use importance sampling to sequentially approximate the posterior distribution. In the context of state-space modelling, we are interested in approximating the posterior density $p(s_t\mid z_{1:t},x)$ of the filtered latent state $s_t$, given a sequence of observations $z_{1:t}$. In this section, we shall assume that the model parameters $x$ are fixed. Approximations of $p(s_{t}\mid z_{1:t},x)$ can be calculated recursively by first approximating $p(s_1\mid z_1,x)$, then $p(s_{2}\mid z_{1:2},x)$ and so forth for $t=1,\ldots,T$. At time $t$ the posterior of the filtered state is 
\begin{equation}
  \label{eq:6}
  p(s_t\mid z_{1:t},x) \propto  g_x(z_t\mid s_t) \int f_x(s_t\mid s_{t-1}) p(s_{t-1}\mid z_{1:t-1},x) ds_{t-1} 
\end{equation}
where $p(s_{t-1}\mid z_{1:t-1},x)$ is the posterior density at time $t-1$.

The posterior at time $t$ can be approximated if we assume that at time $t-1$ we have a set of particles $\{s_{t-1}^{(i)}\}_{i=1}^{N}$ and corresponding normalized weights $\{w_{t-1}^{(i)}\}_{i=1}^N$ which produce a discrete approximation of $p(s_{t-1}\mid z_{1:t-1},x)$. This induces the following approximation to \eqref{eq:6}, 
\begin{equation} \label{eq:1a}
  \hat{p}(s_t\mid z_{1:t},x) \approx c  g_x(z_t\mid s_t) \sum_{i=1}^{N} w_{t-1}^{(i)} f_x(s_t\mid s_{t-1}^{(i)}),
\end{equation}
where $c$ is a normalizing constant. The filtered density, as given above, can be updated recursively by propagating and updating the particle set using importance sampling techniques. The resulting algorithms are called particle filters, see \cite{Doucet2000} and \cite{Cappe2007} for a review.

In this paper the particle approximations of the latent process are created with the auxiliary particle filter of \cite{Pitt1999a}.  
This filter can be viewed as a general filter from which simpler filters are given as special cases \cite[]{Fearnhead2008}. 
The aim is to view the target (\ref{eq:1a}) as defining a joint distribution on the particle at time $t-1$ and the value of a new particle at time $t$. The probability of sampling particle $s_{t-1}^{(i)}$ and $s_t$ is
\begin{equation*}
  c w_{t-1}^{(i)}g_x(z_t\mid s_t)f_x(s_t\mid s_{t-1}^{(i)}).
\end{equation*}
We approximate this with $ \xi_t^{(i)}q(s_t\mid s_{t-1}^{(i)},z_t,x)$, where  $q(s_t\mid s_{t-1}^{(i)},z_t,x)$ is a density function that can be sampled from and $\{\xi_t^{(i)}\}_{i=1}^N$ are a set of probabilities. 
This defines a proposal which we can simulate from by first choosing particle $s_{t-1}^{(i)}$ with probability $\xi_t^{(i)}$, and then, conditional on this, a new particle value, $s_t$, is sampled
from $q(s_t\mid s_{t-1}^{(i)},z_t,x)$. The weight assigned to our new particle is then
\[
\tilde{w}_t=\frac{w_{t-1}^{(i)}g_x(z_t\mid s_t)f_x(s_t\mid s_{t-1}^{(i)})}{\xi_t^{(i)}q(s_t\mid s_{t-1}^{(i)},z_t,x)}.
\]
Details are summarized in Algorithm \ref{alg1}.

The optimal proposal density, in terms of minimizing the variance of the weights \citep{Doucet2000}, is available when $q(s_t\mid s_{t-1}^{(i)},z_t,x) = p(s_t\mid s_{t-1}^{(i)},z_t,x)$ and $\xi_t^{(i)} \propto w_{t-1}^{(i)}p(z_t\mid s_{t-1}^{(i)})$. 
This filter is said to be fully adapted as all the weights $w_{t}^{(i)}$ will equal $1/N$. Generally, it is not possible to sample from the optimal proposal, but alternative proposals can be used which approximate the fully adapted filter.

One of the benefits of using the particle filter is that an estimate for the likelihood $p(z_{1:T}\mid x)$ is given for free from the particle filter output. We can estimate $p(z_t\mid z_{1:t-1},x)$ by
\begin{equation}
\label{eq:8}
  \hat{p}(z_t\mid z_{1:t-1},x) = \sum_{i=1}^N \frac{\tilde{w}_t^{(i)}}{N}, 
\end{equation}
where $\tilde{w}_t^{(i)}$ are unnormalized weights. An unbiased estimate of the likelihood \citep{moral2004feynman} is then
\begin{equation*}
  \hat{p}(z_{1:T}\mid x)=\hat{p}(z_1\mid x) \prod_{t=2}^T \hat{p}(z_t\mid z_{1:t-1},x).
\end{equation*}

\begin{algorithm}
\caption{Auxiliary Particle Filter}          
\label{alg1}                           
\textit{Step 1:} Iteration $t=1$.\\
\quad (a) For $i=1,\ldots,N$, sample particles $\{s_1^{(i)}\}$ from the prior $p(s_1\mid x)$ and set $\tilde{w}_1^{(i)} = p(z_1\mid s_1^{(i)})$. \\
\quad (b) Calculate $C_1=\sum_{i=1}^N \tilde{w}_1^{(i)}$; set $\hat{p}(z_1)=C_1/N$; and calculate normalized weights $w_1^{(i)}=\tilde{w}_1^{(i)} /C_1$ for $i=1,\ldots,N$. \\
\textit{Step 2:} Iteration $t=2,\ldots,T$. Assume a user-defined set of proposal weights $\{\xi_{t}^{(i)}\}_{i=1}^N$ and family of proposal distributions $q(s_t\mid s_{t-1}^{(i)},z_t,x)$. \\
\quad (a) Sample indices $\{k_1,k_2,\ldots,k_N\}$ from $\{1,\ldots,N\}$ with probabilities $\xi_{t}^{(i)}$. \\
\quad (b) Propagate particles $s_t^{(i)} \sim q(\cdot\mid s_{t-1}^{(k_i)},z_{t},x)$. \\
\quad (c) Weight particles $\tilde{w}_t^{(i)} = \frac{w_{t-1}^{(k_i)} g_x(z_t\mid s_t^{(i)})f_x(s_t^{(i)}\mid s_{t-1}^{(k_i)})}{\xi_t^{(k_i)}q(s_t^{(i)}\mid s_{t-1}^{(k_i)},z_t,x)}$ and calculate $C_t=\sum_{i=1}^N\tilde{w}_t^{(i)}$. \\
\quad (d) Obtain an estimate of the predictive likelihood, $\hat{p}(z_t\mid z_{1:t-1},x) = C_t/N $, and calculate normalized weights $w_t^{(i)}=\tilde{w}_t^{(i)}/C_t$ for $i=1,\ldots,N$.
\end{algorithm}

Implementing the particle Langevin algorithm requires an approximation of the gradient of the log posterior density $\nabla \log \pi(x)$, where $\nabla \log \pi(x) = \nabla \log p(z_{1:T}\mid x) + \nabla \log p(x)$. As outlined in the Section \ref{sec:particle-filtering} we can use the \cite{Poyiadjis2011} algorithm to approximate the gradient, however, the variance of this approximation increases quadratically with $t$. An alternative method proposed by \cite{Nemeth2013} has been shown to produce estimates of the gradient with only linearly increasing variance. We shall use this method to create the particle Langevin proposal, details of which are as follows.

For each particle at a time $t-1$, there is an associated path, defined by tracing the ancestry of each particle back in time. With slight abuse of notation denote this path by $s_{1:t-1}^{(i)}$. We can thus associate with particle $i$ at time $t-1$ a value $\alpha_{t-1}^{(i)} = \nabla\log p(s_{1:t-1}^{(i)},z_{1:t-1}\mid x)$. These values can be updated recursively. Remember that in step 2(b) of Algorithm \ref{alg1} we sample $k_i$, which is the index of the particle at time $t-1$ that is propagated to produce the $i\mbox{th}$ particle at time $t$. Thus we have
\begin{equation} \label{eq:3}
   \alpha_t^{(i)} =  \alpha_{t-1}^{(k_i)} + \nabla \log g_x(z_t\mid s_t^{(i)}) + \nabla \log f_x(s_t^{(i)}\mid s_{t-1}^{(k_i)}).
\end{equation}

The main idea behind the \cite{Nemeth2013} approach is to use kernel density estimation to replace each discrete $\alpha_{t-1}^{(i)}$ value by a Gaussian distribution:
\begin{equation} \label{eq:4}
 \alpha_{t-1}^{(i)} \sim \mathcal{N}(m_{t-1}^{(i)},V_{t-1}).
\end{equation}
The mean of this distribution is obtained by shrinking $\alpha_{t-1}^{(i)}$ towards the mean of $\alpha_{t-1}$,
\[
 m_{t-1}^{(i)}=\zeta\alpha_{t-1}^{(i)}+(1-\zeta)\sum_{i=1}^N w_{t-1}^{(i)}\alpha_{t-1}^{(i)}.
\]
Here $0<\zeta<1$ is a user-defined shrinkage parameter. The idea of this shrinkage is that it corrects for the increase in variability introduced through the kernel density estimation of \cite{West1993}. For a definition of $V_{t-1}$ see \cite{Nemeth2013}, however, its actual value does not affect the following details.

The resulting model for the $\alpha_{t}$'s, including their updates (\ref{eq:3}), is linear Gaussian. Hence
we can use Rao--Blackwellization to avoid sampling $\alpha_{t}^{(i)}$, and instead calculate the parameters of the kernel (\ref{eq:4}) directly. This gives the following recursion for the means,
\begin{eqnarray}
  \label{eq:21}
m_t^{(i)} &=& \zeta m_{t-1}^{(k_i)}+(1-\zeta)\sum_{i=1}^N w_{t-1}^{(i)} m_{t-1}^{(i)} + \nabla \log g_x(z_t\mid s_t^{(i)}) + \nabla \log f_x(s_t^{(i)}\mid s_{t-1}^{(k_i)}). \nonumber
\end{eqnarray}

The final score estimate depends only on these means, and is
\[ \nabla \log \hat{p}(z_{1:t}\mid x)=\sum_{i=1}^N w_t^{(i)} m_{t}^{(i)}.\]

See Algorithm \ref{alg6} for a summary.

\begin{algorithm}
\caption{Rao-Blackwellized Kernel Density Estimate of the Score Vector}          
\label{alg6}                           
 Add the following steps to Algorithm \ref{alg1}. \\
 \textit{Step 1:} \\
 (c) Set $\nabla \log \hat{p}(z_{1}\mid x)= \nabla \log g_x(z_1\mid s_1^{(i)}) + \nabla \log \mu_x(s_1^{(i)})$. \\
 \textit{Step 2:} \\
 (e) For $i=1,\ldots,N$, calculate \\
\vspace{-0.5cm}
\[
m_t^{(i)} = \zeta m_{t-1}^{(k_i)}+(1-\zeta)\sum_{i=1}^N w_{t-1}^{(i)} m_{t-1}^{(i)} + \nabla \log g_x(z_t\mid s_t^{(i)}) + \nabla \log f_x(s_t^{(i)}\mid s_{t-1}^{(k_i)}). 
\]
 (f)  Update and store the score vector 
\[
 \nabla \log \hat{p}(z_{1:t}\mid x)= \sum_{i=1}^N w_t^{(i)}m_{t}^{(i)}. 
\]
\end{algorithm}

When $\zeta=1$ the recursion simplifies to the method given by \cite{Poyiadjis2011}, where the variance of the score estimate will increase quadratically with $t$. The use of a shrinkage parameter $\zeta<1$ alleviates the degeneracy problems that affect the estimation of the score and significantly reduces the estimate's variance.
As a rule of thumb, setting $\zeta=0.95$ produces reliable estimates and we shall use this tuning for all examples in the Section \ref{sec:particle-filtering}. Decreasing $\zeta$ leads to a decrease in variance, but at the cost of increasing the bias in the estimate of the gradient. \cite{Nemeth2013} have shown that reliable results can be obtained for a wide range of $\zeta$ and that values in the range $0.5<\zeta<0.99$ work particularly well.  

\section{Negative $K_{**}$} \label{App:K**}

The $K_{**}$ term, that appears in the acceptance rate for regime (2) can be negative. 
This term depends on the interaction between the bias in our estimate of the gradient
and the curvature of the posterior. A negative value corresponds to a case where the bias in our estimate of the gradient is beneficial and actually improves the mixing of the algorithm. Intuitively, for these
cases, the bias is correcting for the error in the Euler discretization of the Langevin diffusion that is used to obtain the Metropolis-adjusted Langevin proposal. A negative $K_{**}$ value can lead 
to the counter-intuitive situation 
where increasing the step-size can sometimes increase the acceptance rate. 

To see how this happens, we present a simple example. We assume that the target distribution has independent and identically distributed standard Gaussian components, and that
the bias of the estimate of the gradient for the component of interest is $b(x)/n^{1/3}=-x/n^{1/3}$. So as to emphasize the effect that the bias is having, we will consider the case where $\sigma^2=\tau^2=0$, so the
likelihood is estimated without error, and the only error in the gradient is due to the bias. We are considering regime (2), where $\kappa=1/3$.

Simple calculations give $K=1/4$, $K_{*}^2=1$ and $K_{**}=-1/4$. The limiting acceptance rate is thus
\[
 \alpha^{(2)}(\ell,0)=2\Phi\left\{-\frac{1}{2}\left(\ell^6/16-\ell^4/2+\ell^2 \right) \right\}.
\]
This limiting acceptance rate is equal to 1 either when $\ell=0$, or when $\ell=2$. 

The Langevin dynamics for the component of interest are defined by the stochastic differential equation
\[
 \mbox{d}X_t=-\frac{1}{2}X_t\mbox{d}t+\mbox{d}B_t.
\]
The standard Langevin algorithm will propose, using an Euler approximation, 
\begin{equation} \label{eq:A1}
 Y=x\left(1-\lambda^2/2\right)+\lambda Z,
\end{equation}
where $x$ is the current value of the chain, and $Z$ is an independent standard Gaussian random variable. The particle Langevin algorithm will have proposal
\begin{equation} \label{eq:A2}
 Y=x\left(1-\lambda^2/2-\frac{\lambda^2}{2n^{1/3}} \right)+ \lambda Z,
\end{equation}
where the difference is due to the bias in the estimate of the gradient. 

Now it is straightforward to show that if  $\lambda<1$ a proposal of the form
\begin{equation} \label{eq:A3}
 Y=x\left(1- \lambda^2 \right)^{1/2}+ \lambda Z
\end{equation}
will have acceptance rate of 1, as this is the true transition density of the Langevin dynamics over a time-step of size $-\log(1-\lambda^2)$. 

For our asymptotic regime (2) we have $\lambda=\ell n^{-1/6}$ and we let $n\rightarrow \infty$. We can expand the coefficient of $x$ in (\ref{eq:A3}) to give
\[
 \left(1- \lambda^2 \right)^{1/2}=1-\frac{1}{2}\lambda^2-\frac{1}{8}\lambda^4+O\left(\lambda^6\right).
\]
The standard Langevin proposal (\ref{eq:A1}) is the same as the ideal
proposal (\ref{eq:A3}) up to order $\lambda^2$. By comparison there is
a 
local maximum of our limiting acceptance rate at $\ell=2$, and, with
this scaling, the particle Langevin proposal (\ref{eq:A2}) is better as it is the same as the ideal proposal (\ref{eq:A3}) up to order $\lambda^4$.

\section{Empirical analysis of assumptions pertaining to theoretical results}
\label{sec:empir-analys-assumpt}

Our theoretical results are posited on a number of simplifying
assumptions. Some, such as the shape of the target and the
independence between position and the distribution of the noise in the
log-target are discussed at the start of Section \ref{sec:optim-scal-results}. Others, such
as the asymptotic distribution of the particle filter estimates, are
based on previous theory \citep{Berard2014} and have been
investigated previously
\cite[e.g.,][]{Sherlock:2015,Doucet:2015}. Others pertain to the
estimates of the gradient of the log-target and are entirely new. In this
section, we verify that many of these assumptions hold approximately for the two examples in
our simulation study.

Our theoretical results also show three possible regimes, with the
final regime, where the effect of the error in the gradient is negligible, being the most desirable. We describe diagnostics that
relate to the regime and we use these to show that both of our simulation
studies are in the desirable regime (3).

\subsection{Noise in the log posterior density}
\label{sec:noise-log-posterior}

Theorems \ref{thrm.main} and \ref{sec:scaling-grads} both assume that
the distribution of the noise in the log posterior density is independent of
the position in the parameter value, $x$. Corollary
\ref{sec:opt-acc-rate}, and our maximin procedure, specify further that
the noise is Gaussian \eqref{eqn.SARa}
with a variance that is inversely proportional to the number of
particles. 
These three assumptions have been made before
\citep{Doucet:2015,Sherlock:2015,Pitt2012}; the first of
them, in particular, is unlikely to hold in practice but has been found
to hold approximately. The second and third are suggested by
particle filter theory
\citep{moral2004feynman,Berard2014}. We now check these assumptions for our simulation study examples.
 
Figure \ref{fig:var_post} shows a histogram of the variance of the
noise in the log posterior density evaluated at $100$ points sampled at
random from the posterior. It can be seen that the variance 
fluctuates by about half an order of magnitude either side of a
central value. \cite{Sherlock2014} shows that for
random walk-based algorithms a key
quantity of interest, the optimal scaling, is robust to changes in the
global distributions of $V^{n}$ and $W^{n}$; Figure \ref{fig:contours} suggests
a similar robustness for the particle Langevin algorithm. In moderate to high
dimensions the particle Langevin algorithm can require many iterations to
traverse the posterior. Provided the variance in the noise changes
sufficiently slowly, the variance will appear to be approximately
constant for many consecutive iterations; thus, tuning to the 
optimal scaling that would apply to the current variance if it were
global should be close to
optimal locally. Since the optimal scaling is robust to the variance it seems plausible, that, as suggested by our empirical findings, guidance from
our theory may be robust to (sufficiently slow) local variations in the distributions of $V^{n}$ and $W^{n}$.
\begin{figure}[t!]
  \centering
  \subfigure{\includegraphics[width=0.45\textwidth]{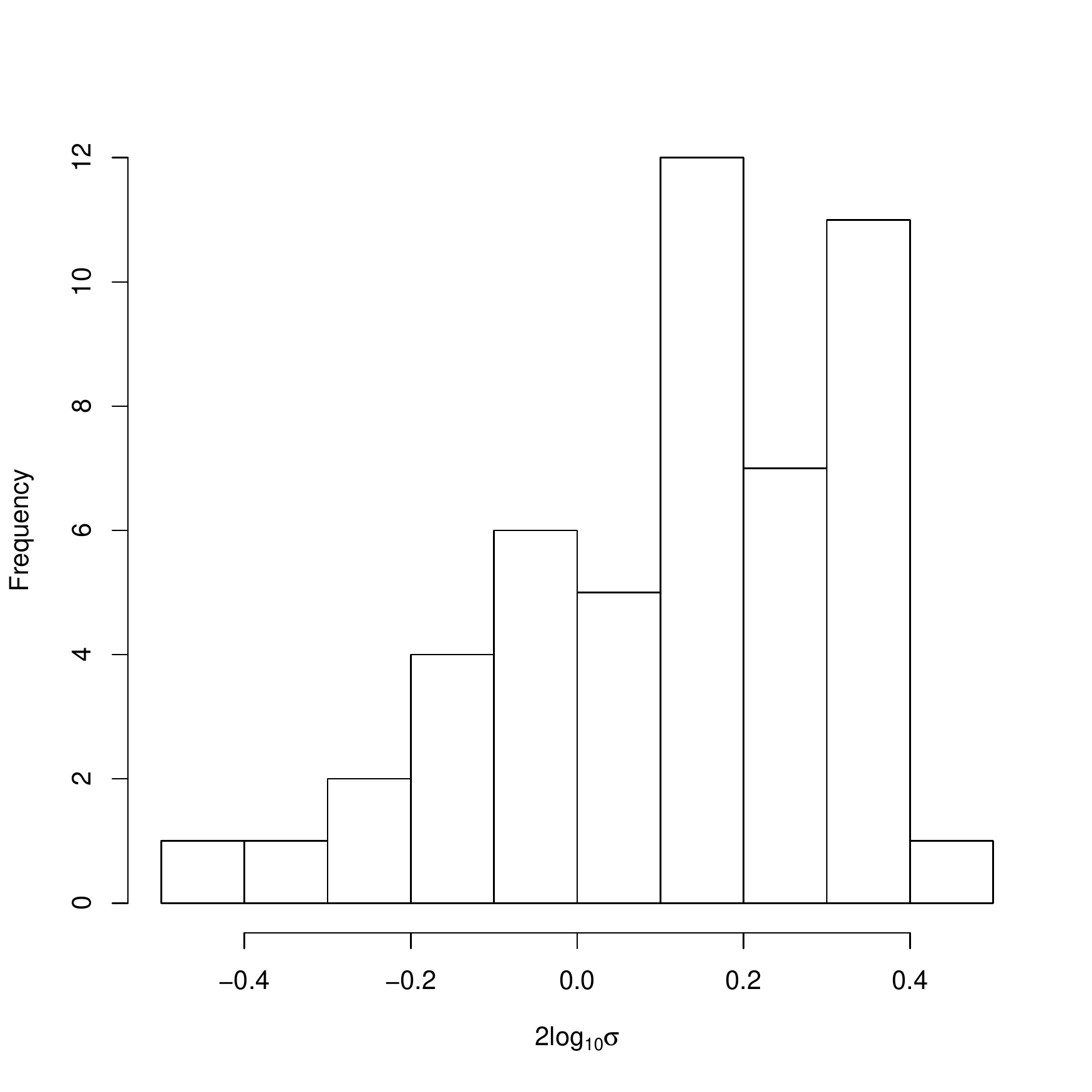}}
  \subfigure{\includegraphics[width=0.45\textwidth]{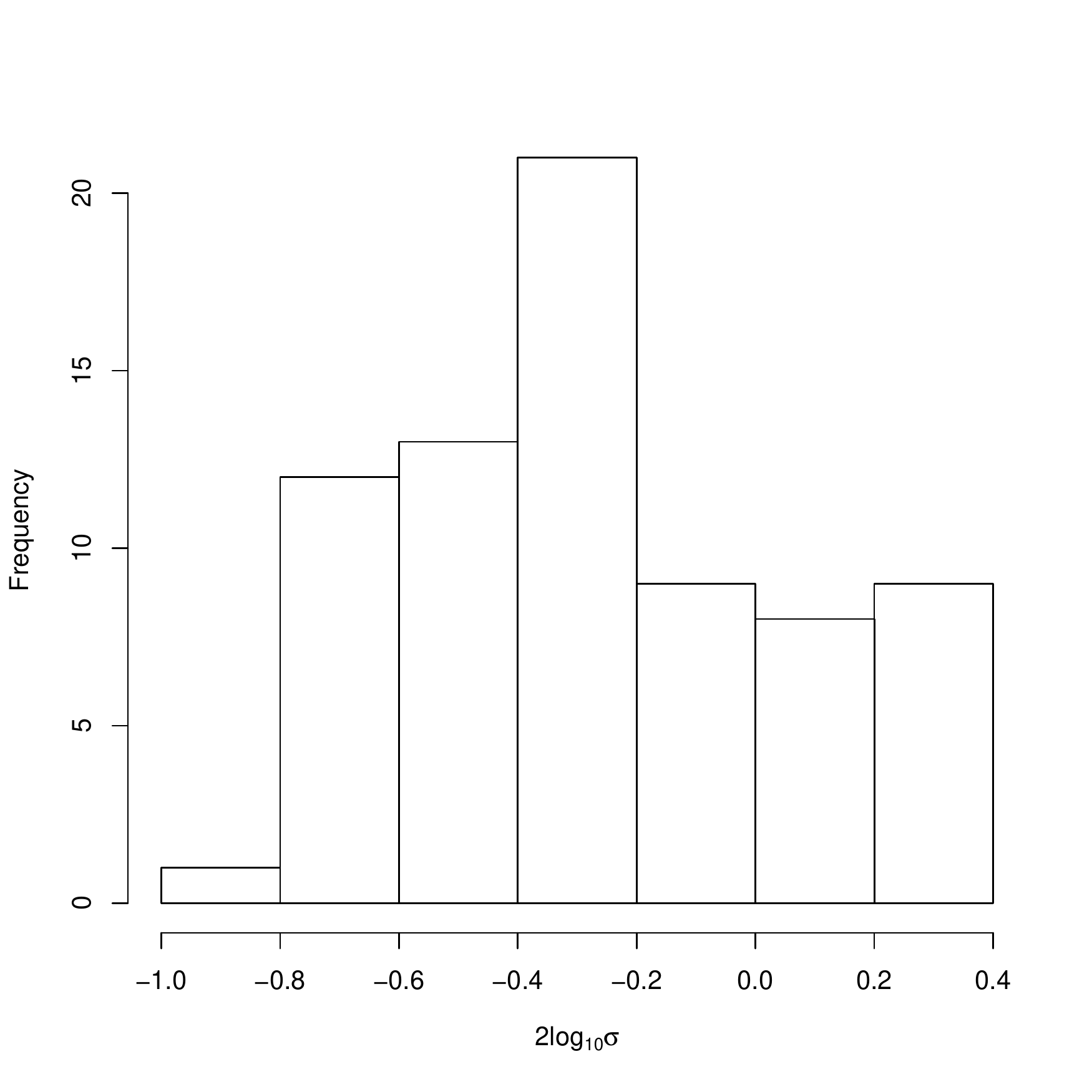}}
  \caption{Base 10 logarithm of the variance of the log posterior density at random points in the posterior for the linear Gaussian model (left panel) and the mixture of experts model (right panel).}
  \label{fig:var_post}
\end{figure}

Figure \ref{fig:post_noise} shows, for each of our two examples,
kernel density estimates of the log posterior density based on $500$
point estimates at each of two points sampled from the
posterior. The noise in the log posterior density is, at least approximately, Gaussian.
This is an important check as the theory that predicts a Gaussian distribution is based upon the use of a large
number of particles, but for the linear Gaussian and mixture of experts
models we needed respectively only $20$ and $100$ particles. 

\begin{figure}[t!]
  \centering
  \subfigure{\includegraphics[width=0.45\textwidth]{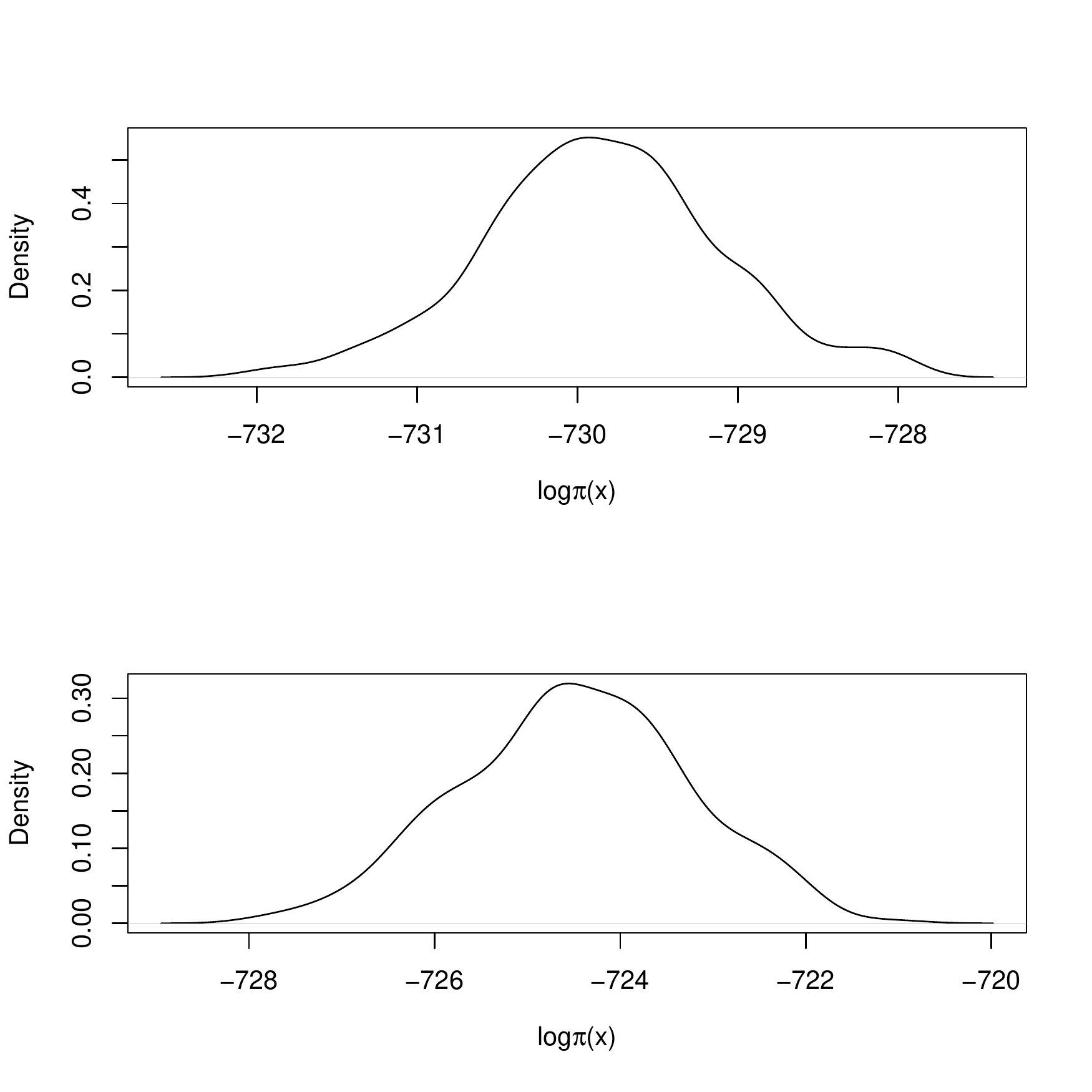}}
  \subfigure{\includegraphics[width=0.45\textwidth]{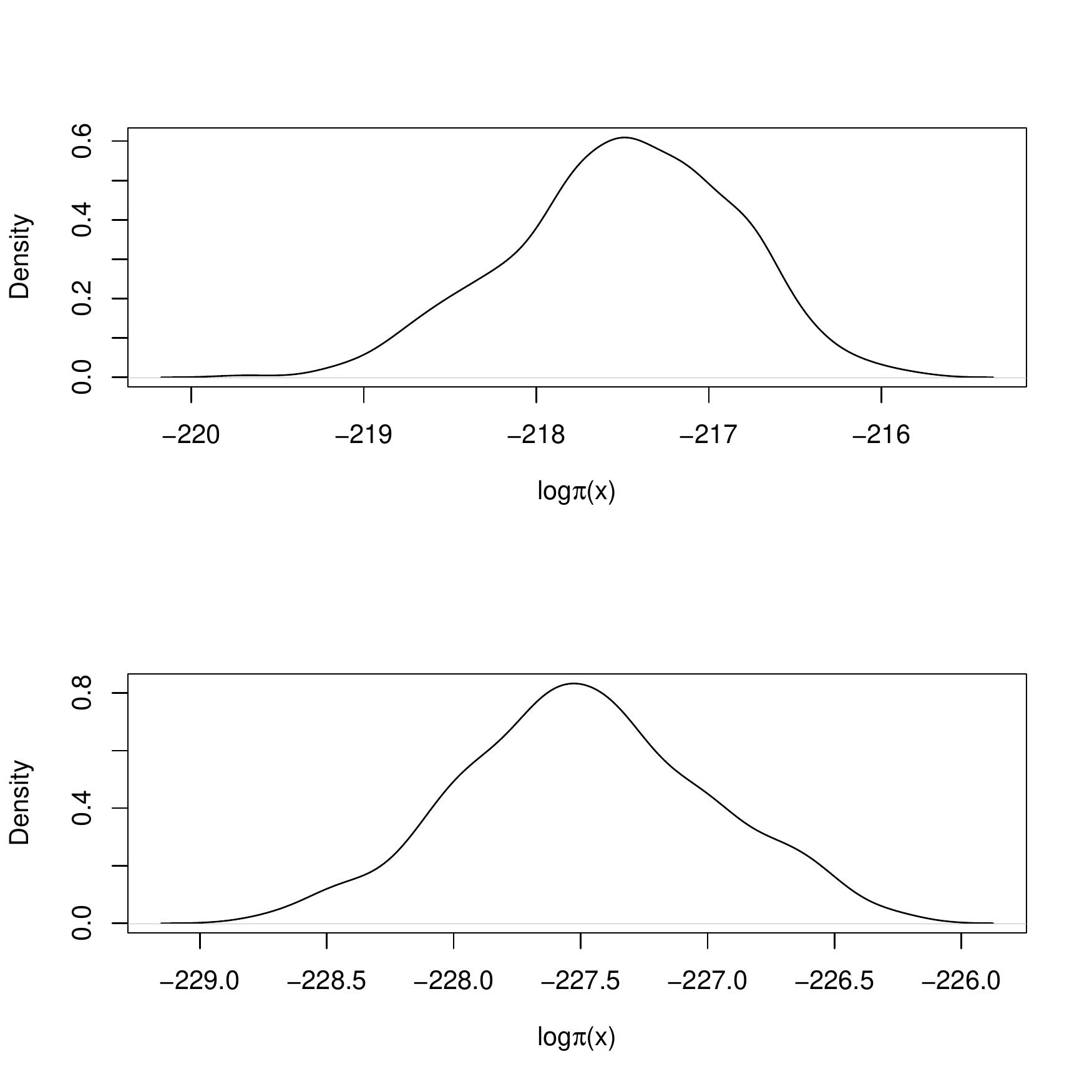}}
  \caption{Empirical log posterior density, taken at two random points in the posterior, for the linear Gaussian model (left panels) and the mixture of experts model (right panels).}
  \label{fig:post_noise}
\end{figure}

For each example, Figure \ref{fig:particles_vs_var} plots an
estimate of the logged-variance (obtained using
$500$ repeated estimations of the log posterior density for each number of particles) evaluated at the same random point in
the posterior against the logged number of particles. The straight line has gradient $-1$  and shows that the variance is indeed inversely proportional to the  number of particles.

\begin{figure}[t!]
  \centering
  \subfigure{\includegraphics[width=0.45\textwidth]{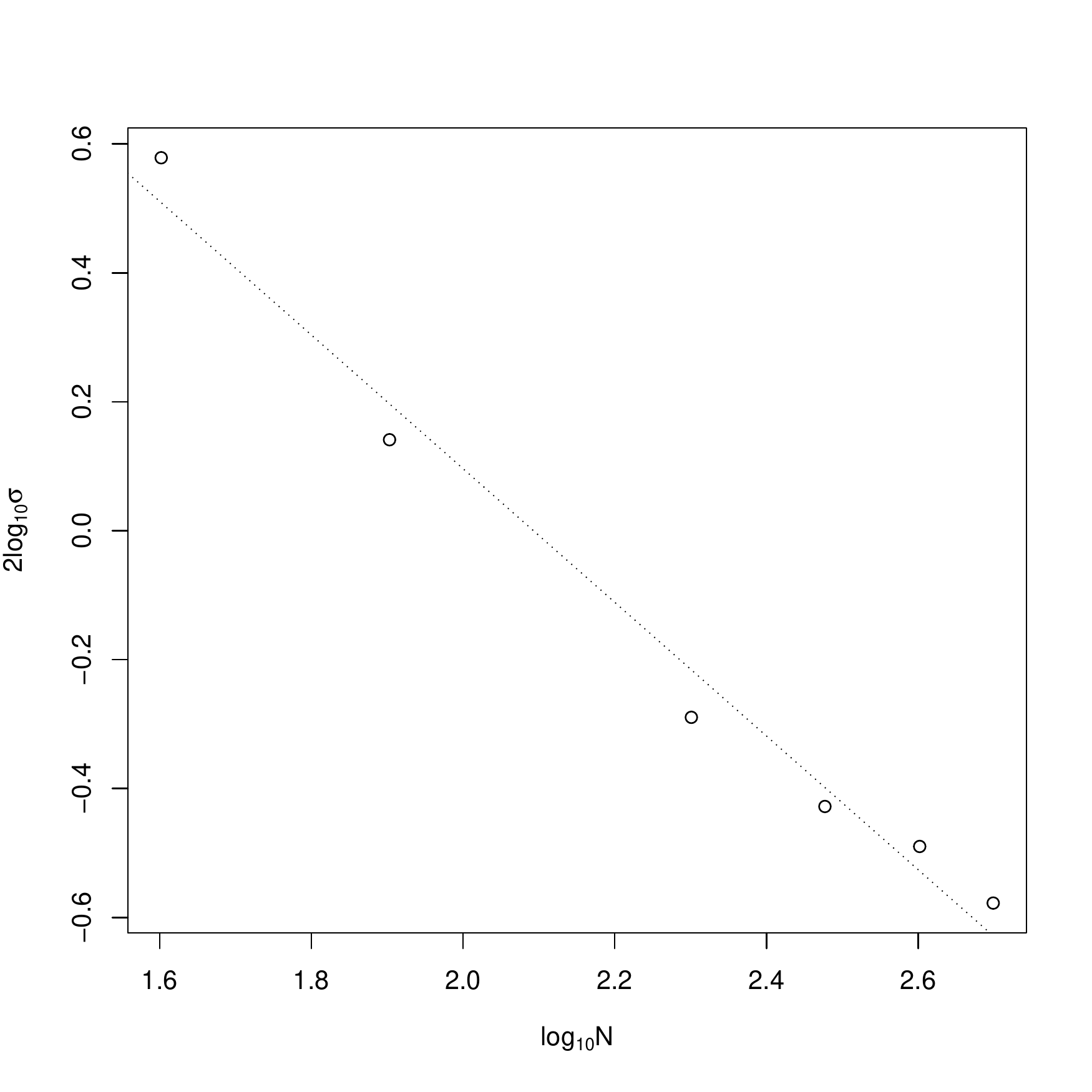}}
  \subfigure{\includegraphics[width=0.45\textwidth]{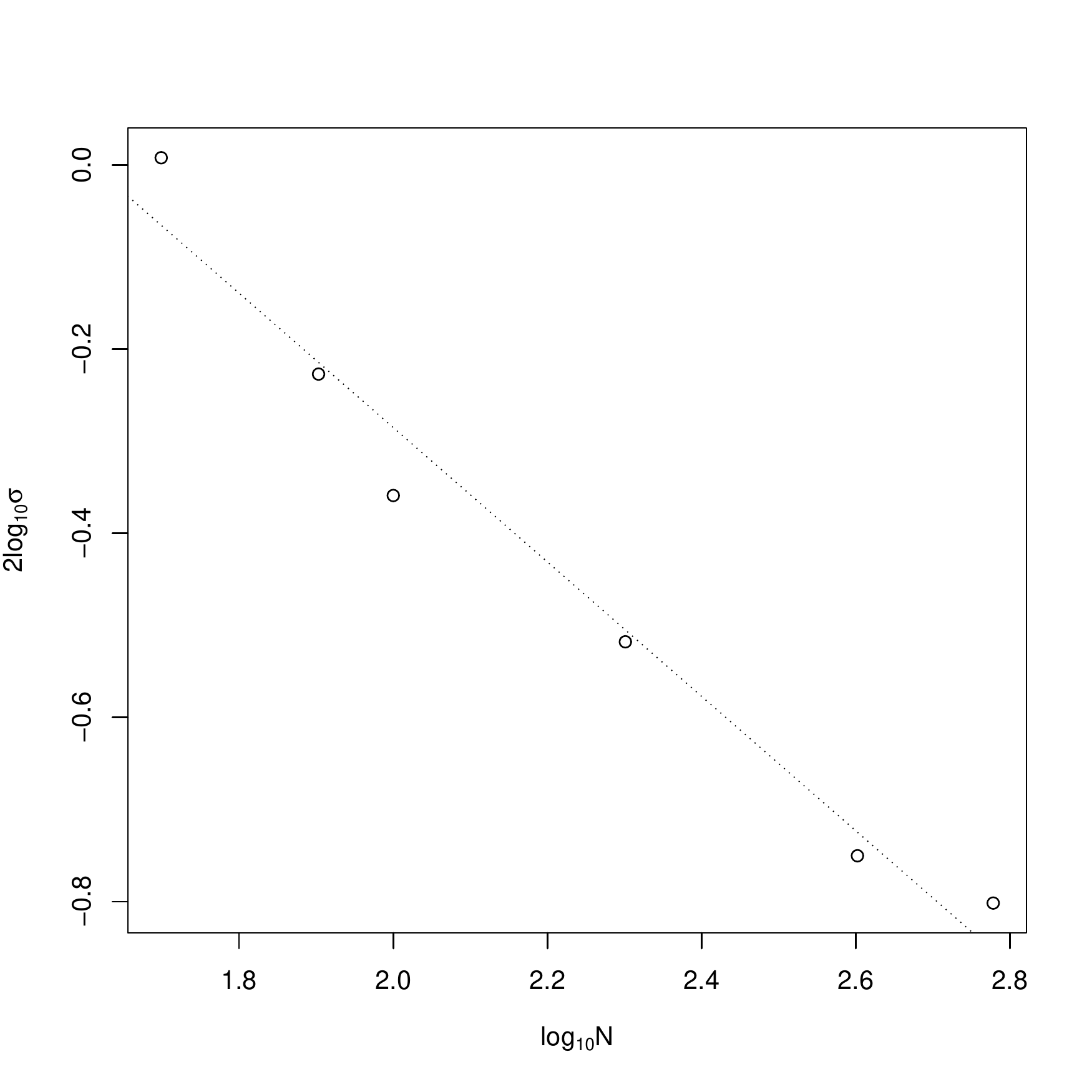}}
  \caption{Number of particles against the empirical variance, taken
    at a random point in the posterior, for the linear Gaussian model
    (left panel) and the mixture of experts model (right panel). The
    diagonal lines have slope $-1$.}
  \label{fig:particles_vs_var}
\end{figure}

\subsection{Noise in the estimate of the gradient}

Theorem \ref{sec:scaling-grads} allows for an  
 error in the 
estimate of a given component of the gradient in the
log posterior density. The variance of this error is assumed to be independent
of position and the error is assumed to be independent of the error in
the estimate of the log posterior density. This independence is by no means
certain since both estimates are created from the same run of a
particle filter.

To test these assumptions, in each of our two scenarios in Section \ref{sec:particle-filtering},
the linear Gaussian and mixture of experts examples, we sampled $100$
points independently from the posterior. For each of these points we ran the
particle filter $500$ times, creating $500$ estimates of the
log posterior density and $500$ estimates of the gradient of the
log posterior density. 

Figure \ref{fig:post_grad} plots, for one of these points in the
posterior, the estimate of the log posterior density against the
first and second components of the estimate of 
$\nabla \log\pi$. This lack of any visible pattern was repeated over the remaining 4 and 8 components of the linear Gaussian and mixture of experts models, respectively, and also over other points in the posterior. 
\begin{figure}[t!]
  \centering
  \subfigure{\includegraphics[width=0.45\textwidth]{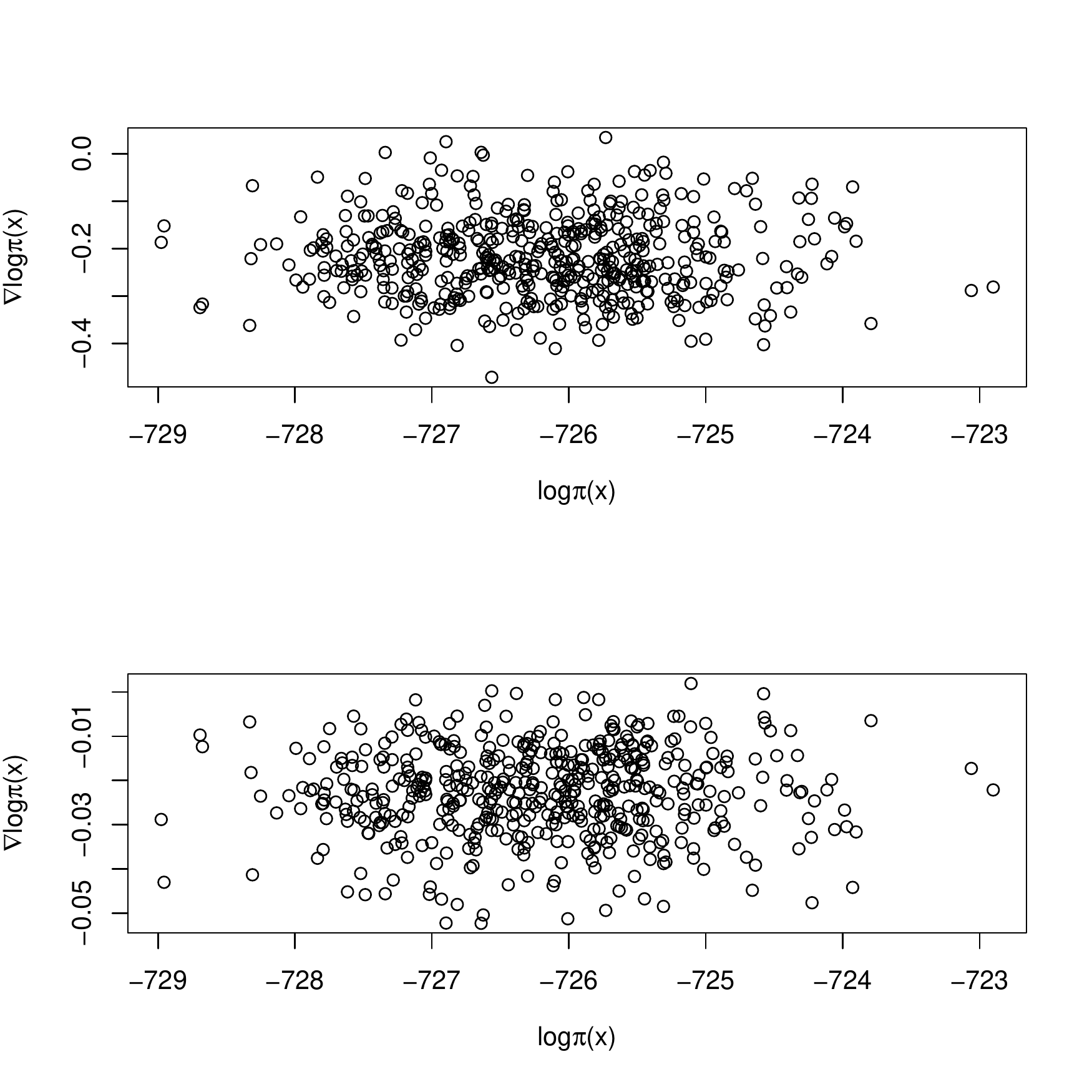}}
  \subfigure{\includegraphics[width=0.45\textwidth]{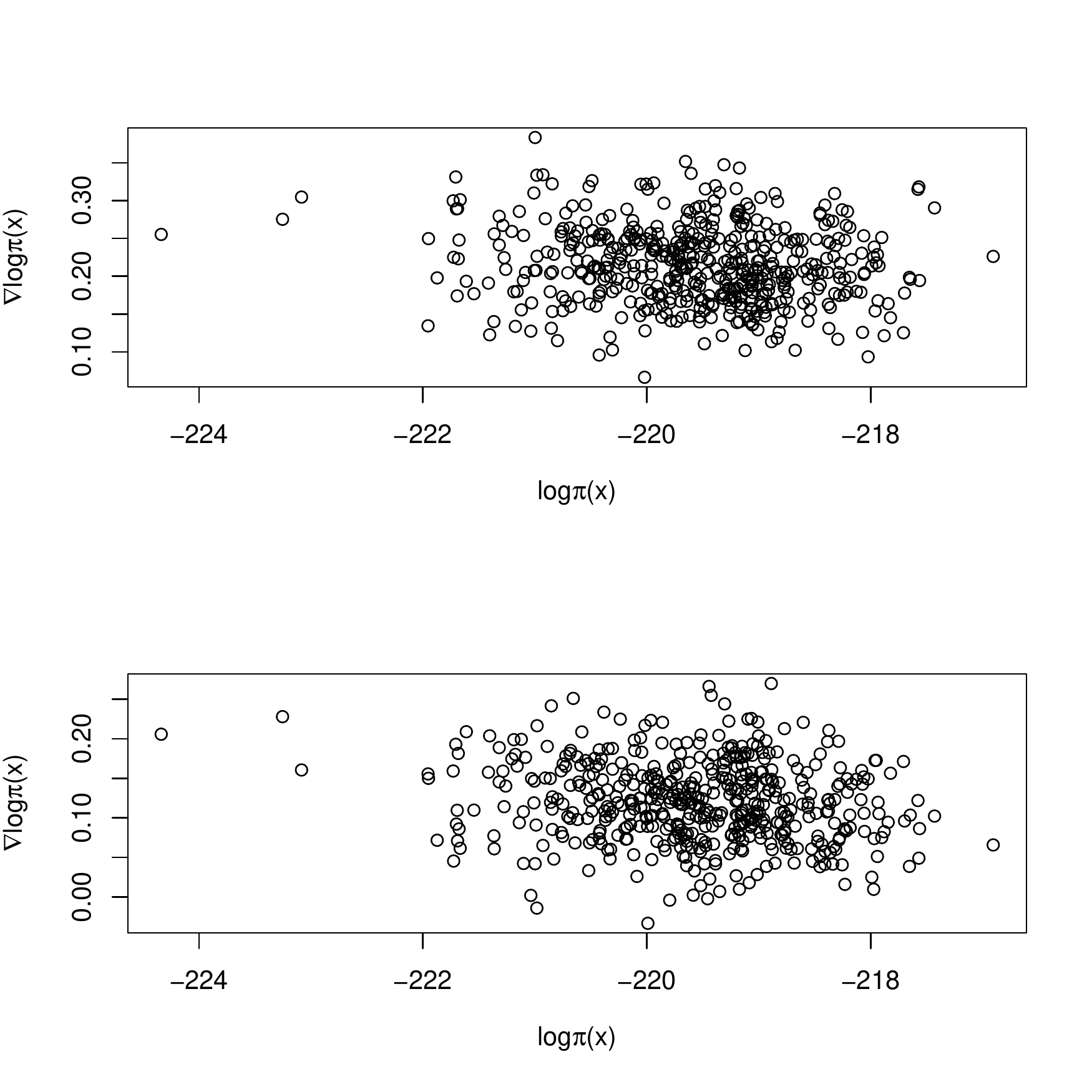}}
  \caption{Estimates of the log posterior density against the first and second components of the gradient of the log posterior density for the linear Gaussian model (left panel) and mixture of experts model (right panel).}
  \label{fig:post_grad}
\end{figure}

Figure \ref{fig:hist_var} presents a histogram of the variance of these
estimates in the gradient for each of the $6$ parameters in the linear Gaussian model. It shows that the variation in this variance
across the posterior is typically of an order of magnitude or less. 

\begin{figure}[t!]
  \centering
  \includegraphics[scale=0.5]{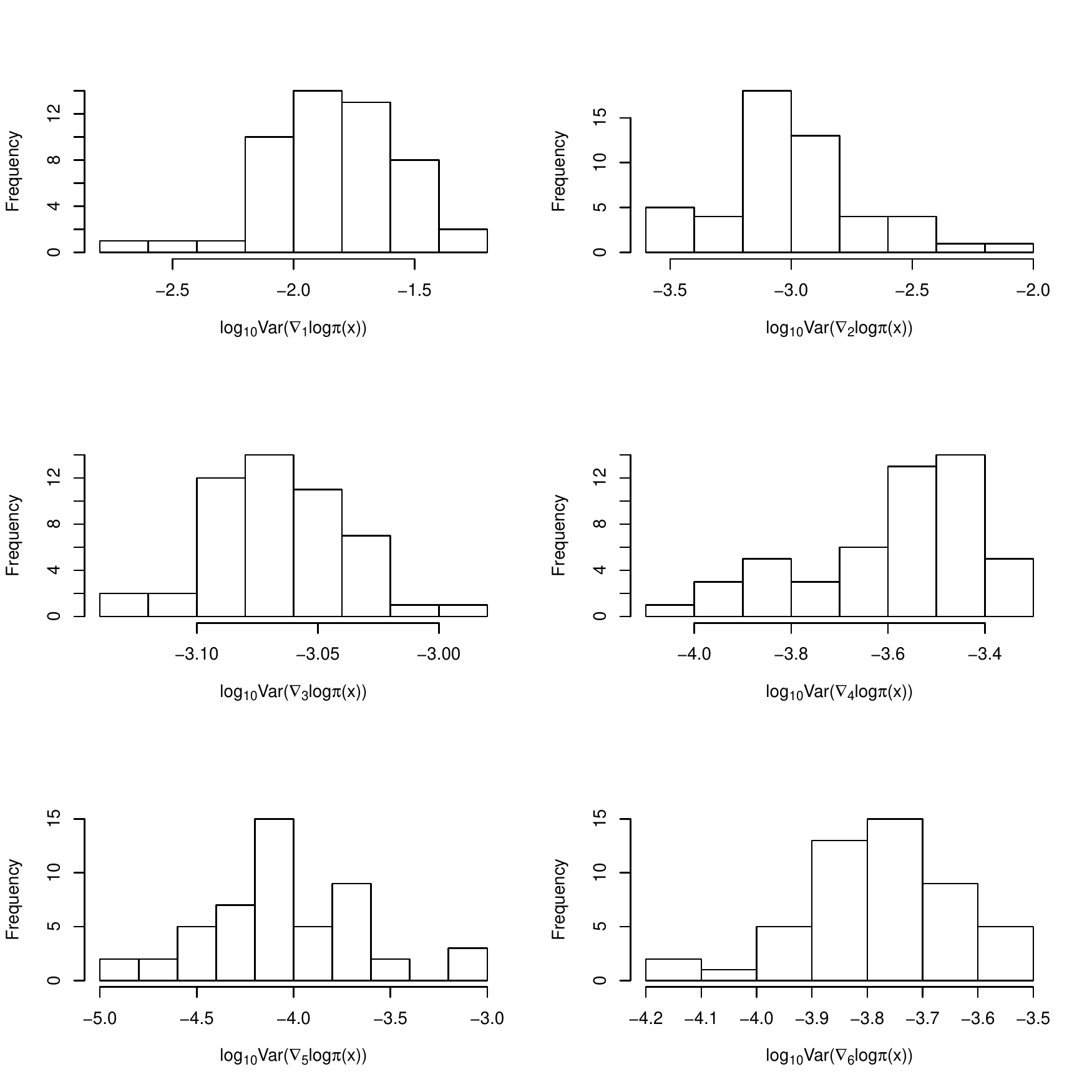}
  \caption{Histogram of log base 10 variances for each component of the gradient of the log posterior density for the linear Gaussian model.}
  \label{fig:hist_var}
\end{figure}

Figure \ref{fig:grad_noise} presents kernel density estimates of the
distribution of the noise in the estimates of the first two components
of the gradient in the log posterior density. The shapes suggest
that this density has light tails, in line with our assumption of
finite moments \eqref{eqn.mom.U}. Additionally, although we did not require this,
it is interesting that the noise in the gradient appears, at least
approximately, to be Gaussian.

\begin{figure}[t!]
  \centering
  \subfigure{\includegraphics[width=0.45\textwidth]{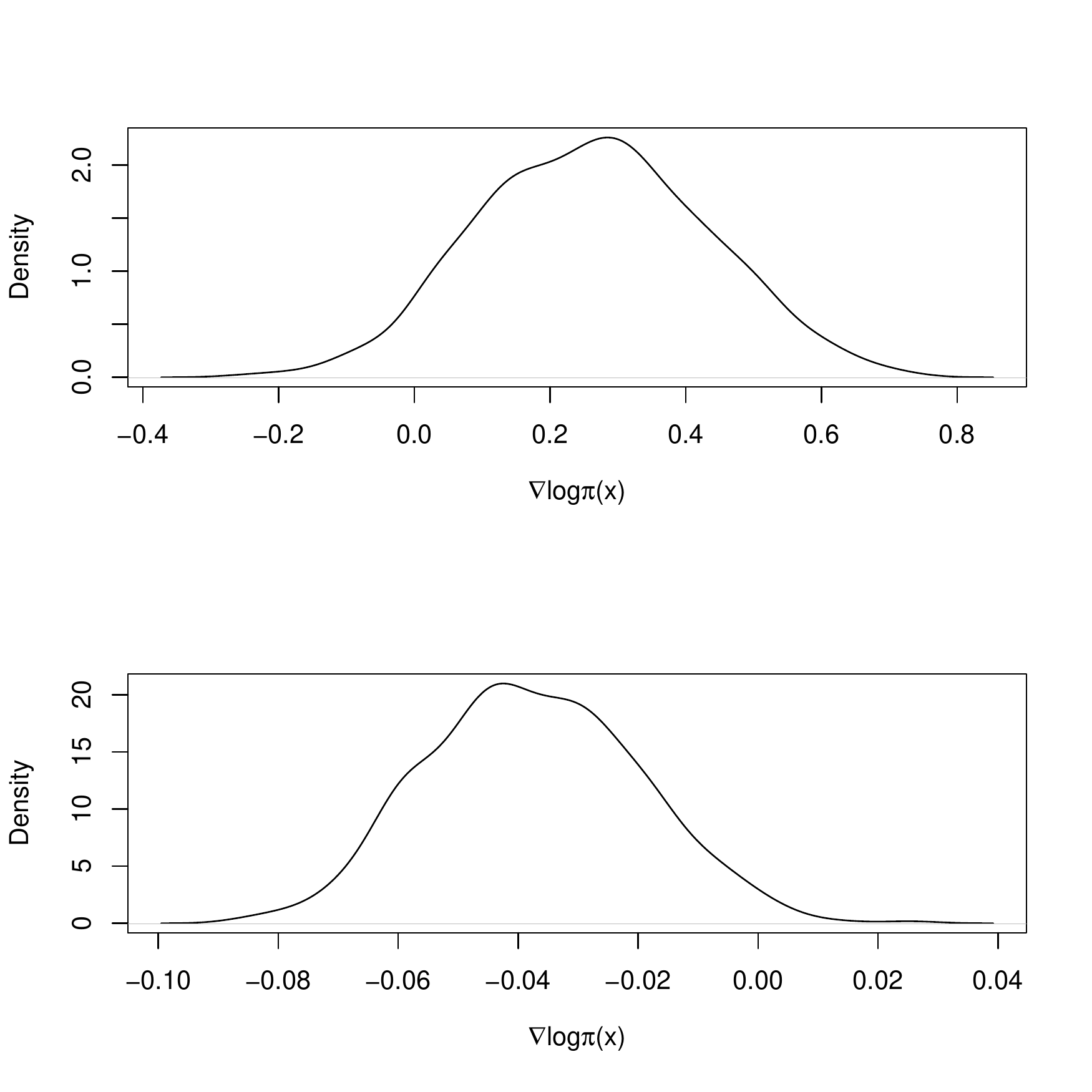}}
  \subfigure{\includegraphics[width=0.45\textwidth]{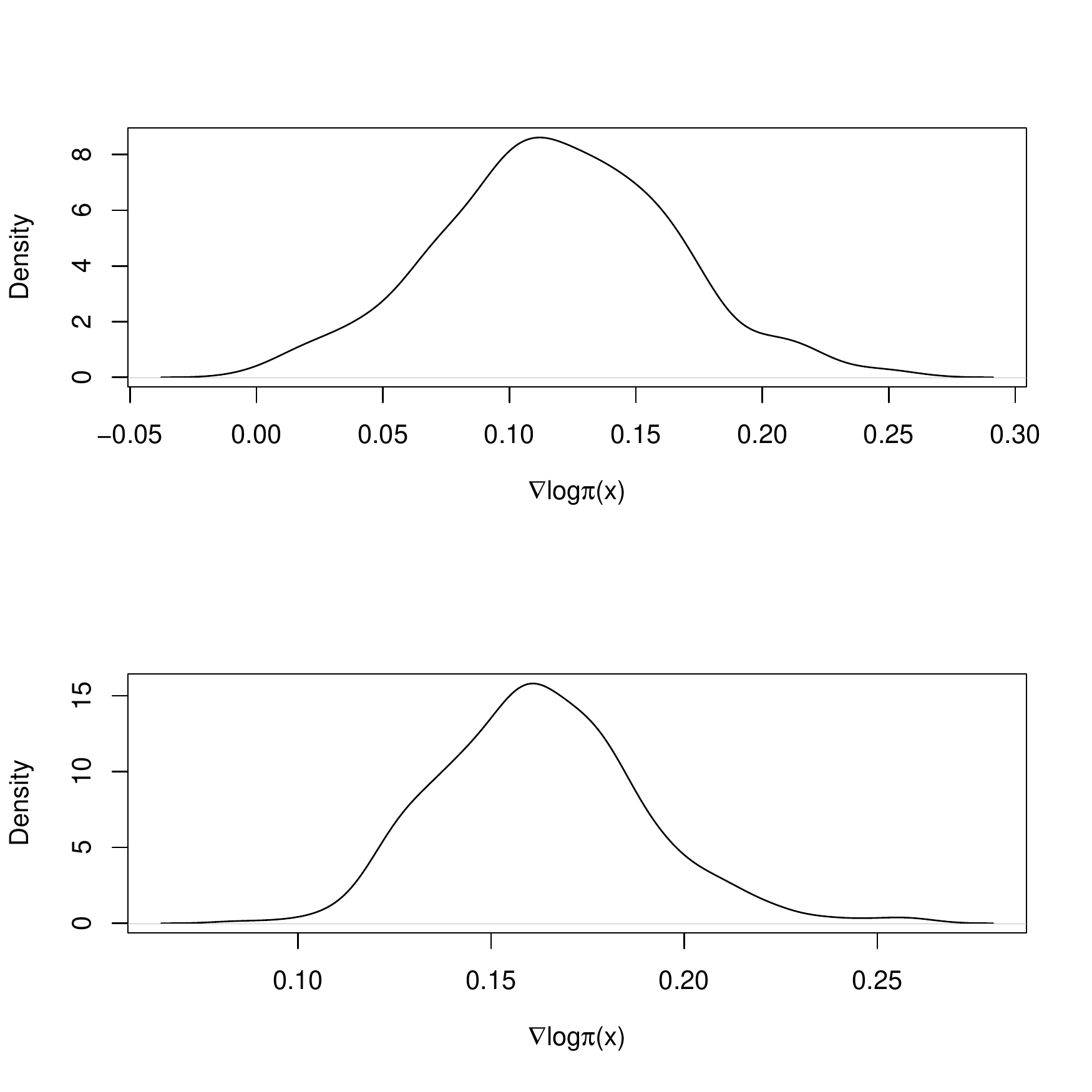}}
  \caption{Empirical density of the first and fourth components of the gradient of the log posterior density, taken at a random point in the posterior, for the linear Gaussian model (left panel) and the mixture of experts model (right panel).}
  \label{fig:grad_noise}
\end{figure}

\subsection{Regime diagnostics}
\label{sect.regime.diags}
Suppose for simplicity that we know precisely the log posterior density at the
current value $x$. We then estimate the log posterior density, $\log \pihat(x')$, at a proposed
value, $x'$. The change in the log posterior density, 
$\Delta:=\log \pihat(x')-\log \pi(x)$ can be split in to three
separate contributions:
\begin{itemize}
\item[$\Delta_A$] The change in the log posterior density that would have resulted if
  we had proposed a new value using the true gradient, $\nabla \log \pi(x)$.
\item[$\Delta_B$] The additional change in the log posterior density because we
  actually used an approximate gradient, $\hat{\nabla} \log \pi(x)$.
\item[$\Delta_C$] The error in the log posterior density at the proposed new value.
\end{itemize}
Throughout Theorem \ref{sec:scaling-grads}, $\Delta_C$ is assumed to have
a variance of $\sigma^2$ which we expect to be $O(1)$. In
regime (1), however $|\Delta_C|\sim|\Delta_B|>>|\Delta_A|$, whereas in Regime (3)
$|\Delta_C|\sim|\Delta_A|>>|\Delta_B|$. In Regime (2) the terms are
all of similar magnitudes.

To be specific, define
\begin{eqnarray*}
  x^* &=& x + \lambda Z + \frac{\lambda^2}{2} \nabla \log \pi(x) \\
  x^{\prime} &=& x + \lambda Z + \frac{\lambda^2}{2} \hat{\nabla} \log \pi(x),
\end{eqnarray*}
where $Z \sim \mathcal{N}(0,I)$. The first proposal is the standard Metropolis-adjusted Langevin proposal where the gradient is known exactly and the second is the particle Langevin proposal. Then
\begin{eqnarray*}
  \Delta_A &=& \log\pi(x^*) - \log\pi(x),\\
  \Delta_B &=& \log\pi(x^\prime) - \log\pi(x^*),\\ 
  \Delta_C &=& \log\pihat(x^\prime) - \log\pi(x^\prime).
\end{eqnarray*}

For each example in Section \ref{sec:particle-filtering}, and for each of $50$ points in the
posterior (each representing a value of $x$), we performed the
following. We ran the particle
filter $50$ times to obtain $50$ estimates, $\hat{\nabla}\log  \pi$, 
and then, for the mixture of experts model,
one further time with a very large number of particles to get a
very accurate estimate of $\nabla \log \pi(x)$ (for the linear Gaussian model this was calculated exactly using a Kalman filter). 
For each of the 50 estimates of $\nabla \log \pi$ we also simulated a vector of Gaussian random variables $Z$. This lead to $50$ pairs of $(x^*,x')$ values.
 For each $x^*$ and $x'$, for the mixture of experts model we ran the
particle filter with a very large number of particles to obtain a very
good estimate of the true log posterior density (for the linear Gaussian model
this was obtained from the Kalman filter), we also ran the particle
filter with $N$ particles (where $N$ is the same as in Section \ref{sec:particle-filtering}) to obtain an
estimate of the log posterior density at $x^\prime$. Thus for each of the
$50$ points we obtained $50$ estimates of
$\Delta_A,\Delta_B ~\mbox{and}~\Delta_C$.

\begin{figure}[t!]
  \centering
  \subfigure{\includegraphics[width=0.45\textwidth]{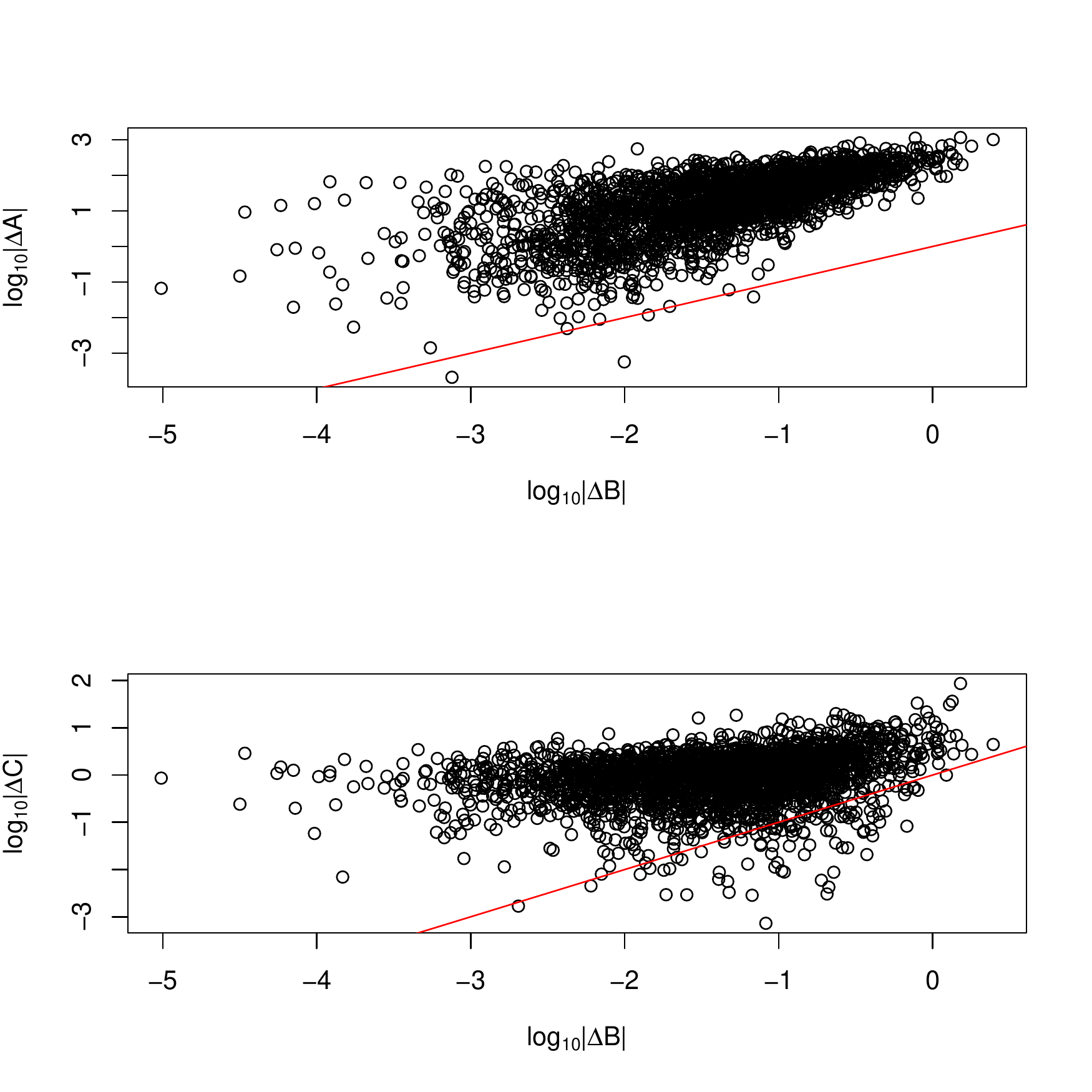}}
  \subfigure{\includegraphics[width=0.45\textwidth]{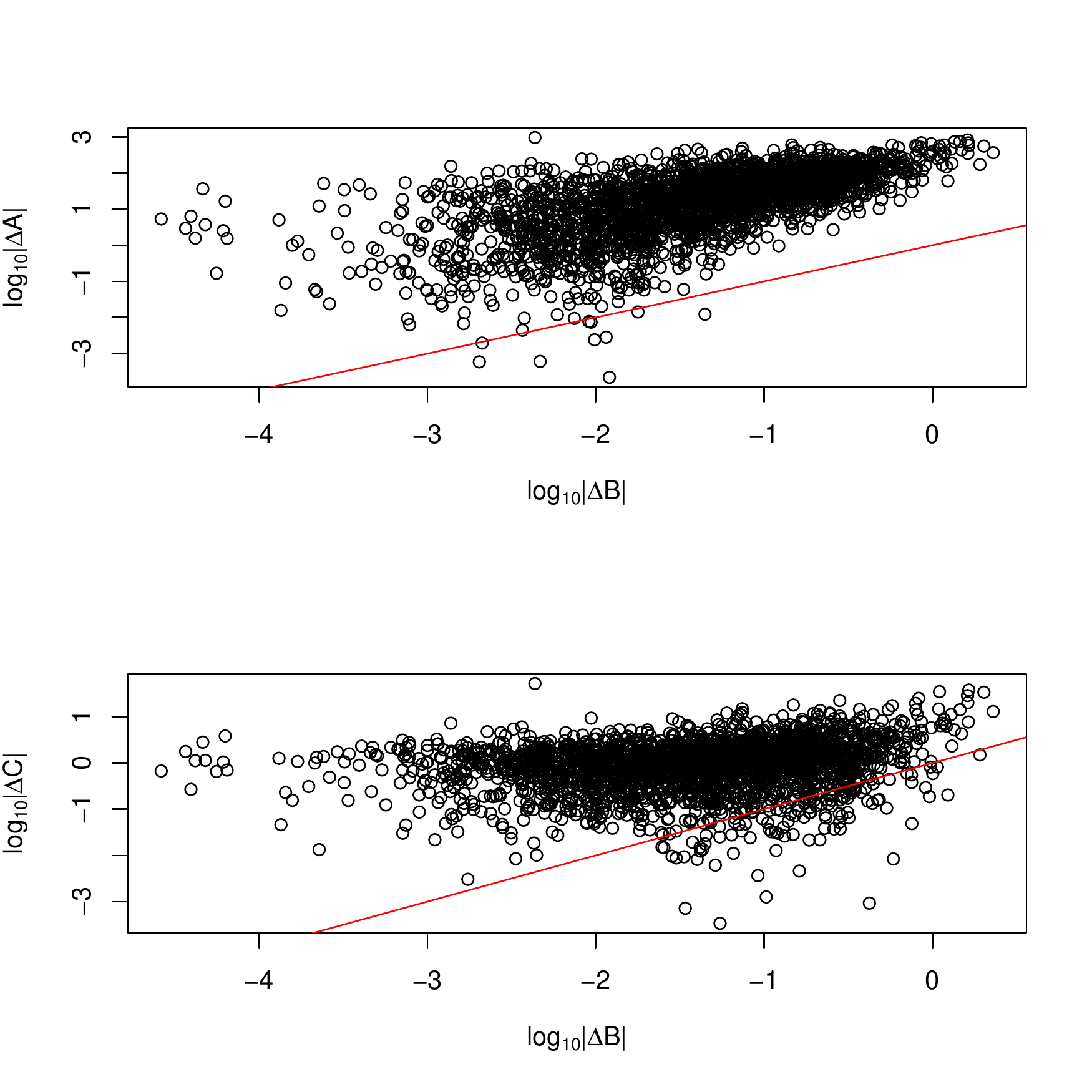}}
  \caption{Regime diagnostics for the linear Gaussian model (left panel) and the mixture of experts model (right panel) in log base 10. The red line in each plot represents equality.}
  \label{fig:deltas}
\end{figure}

Figure \ref{fig:deltas} plots $\Delta A$ and $\Delta_C$ against $\Delta_B$ for
each of the $2,500$ points. It can be seen from the plot that for both examples, $|\Delta_C|\sim|\Delta_A|>>|\Delta_B|$, confirming empirically that we expect to be in regime (3) of Theorem \ref{sec:scaling-grads}.

\end{document}